\documentclass{scrartcl}

\KOMAoptions{paper=letter}

\usepackage[todo]{group}

\usepackage{eucal}
\usepackage{nicefrac}
\usepackage{arydshln}
\usepackage{framed}
\usepackage{bm,bbm,esvect,tabu}
\usepackage{thm-restate}
\usepackage{pgfplots} 
\usepackage{calc}
\usepackage{etoolbox}
\usepgfplotslibrary{colorbrewer,ternary,external}
\pgfplotsset{compat = newest}

\setkomafont{caption}{\footnotesize}
\setkomafont{captionlabel}{\normalsize}

\makeatletter
\def\@makefnmark{% 
  \leavevmode
  \raise.9ex\hbox{\fontsize\sf@size\z@\normalfont\tiny\@thefnmark}}
\makeatother

\newcommand{\prob}[1]{\operatorname{\mathbb P}\left(#1\right)}

\newcommand{\ev}[1]{\operatorname{\mathbb E}\left(#1\right)}

\theoremstyle{remark}

\newcommand{\dist}[1]{\mathrm{dist}\left(#1\right)}

\makeatletter
\newcommand*{\rom}[1]{(\expandafter{\romannumeral #1\relax})}
\makeatother

\title{A Natural Adaptive Process for\\ Collective Decision-Making} 

\author{Florian Brandl\\\normalsize Department of Economics\\[-1ex]\normalsize University of Bonn% 
\and 
Felix Brandt% 
\\\normalsize Department of Computer Science\\[-1ex]
\normalsize Technical University of Munich% 
}

\renewcommand{\epsilon}{\varepsilon}

\newcommand{\p}[1][]{% 
\ifthenelse{\equal{#1}{}}{{R}}{{R_{#1}}}% 
}
\newcommand{\s}[1][]{% 
\ifthenelse{\equal{#1}{}}{{\succ}}{{\succ_{#1}}}% 
}
\newcommand{\srel}[1][]{% 
\ifthenelse{\equal{#1}{}}{\succ}{\succ_{#1}}% 
}

\newcolumntype{P}[1]{>{\centering\arraybackslash}p{#1}}

\newcolumntype{Y}{>{\centering\arraybackslash}X}

\allowdisplaybreaks 

\begin{document}

\maketitle

\begin{abstract}

Consider an urn filled with balls, each labeled with one of several possible collective decisions. Now, let a random voter draw two balls from the urn and pick her more preferred as the collective decision. Relabel the losing ball with the collective decision, put both balls back into the urn, and repeat. 
Once in a while, relabel a randomly drawn ball with a random collective decision. 
We prove that the empirical distribution of collective decisions produced by this process approximates a maximal lottery, a celebrated probabilistic voting rule proposed by Peter C. Fishburn (Rev. Econ. Stud., 51(4), 1984). In fact, the probability that the collective decision in round $n$ is made according to a maximal lottery increases exponentially in $n$.
The proposed procedure is more flexible than traditional voting rules and bears strong similarities to natural processes studied in biology, physics, and chemistry as well as algorithms proposed in machine learning.

\end{abstract}

\section{Introduction}

\label{sec:intro}

The question of how to collectively select one of many alternatives based on the preferences of multiple agents has occupied great minds from various disciplines. Its formal study goes back to the Age of Enlightenment, in particular during the French Revolution, and the important contributions by Jean-Charles de Borda 
and Marie Jean Antoine Nicolas de Caritat, 
better known as the Marquis de Condorcet. Borda and Condorcet agreed that plurality rule---then and now the most common collective choice procedure---has serious shortcomings. This observation remains a point of consensus among social choice theorists and is largely due to the fact that plurality rule merely asks each voter for her most-preferred alternative \citep[see, e.g.,][]{BrFi02a,Lasl11a}.\footnote{For example, plurality rule may select an alternative that an overwhelming majority of voters consider to be the worst of all alternatives.} When eliciting more fine-grained preferences such as complete rankings over all alternatives from the voters, much more attractive choice procedures are available. As a matter of fact, since \citeauthor{Arro51a}'s (\citeyear{Arro51a}) seminal work, the standard assumption in social choice theory is that preferences are given in the form of binary relations that satisfy completeness, transitivity, and often anti-symmetry. Despite a number of results which prove critical limitations of choice procedures for more than two alternatives \citep[e.g.,][]{Arro51a,Gibb73a,Satt75a}, there are many encouraging results \citep[e.g.][]{Youn74a,YoLe78a,BrFi78a,Lasl00a}. In particular, when allowing for randomization between alternatives, some of the traditional limitations can be avoided and there are appealing choice procedures that stand out \citep{Gibb77a,Bran13a,BrBr17a}.

The standard framework in social choice theory rests on a number of rigid assumptions that confine its applicability: there is a fixed set of voters, a fixed set of alternatives, and a single point in time when preferences are to be aggregated; all voters are able to rank-order all alternatives; there is a central authority that collects all these rankings, computes the outcome, and convinces voters of the outcome's correctness, etc. On top of that, computing the outcome of many attractive choice procedures is a demanding task that requires a computer, which can render the process less transparent to voters.\footnote{In some cases, computing the outcome was even shown to be NP-hard, i.e., 
the running time of all known algorithms for computing election winners increases exponentially in the number of alternatives \citep[see, e.g.,][]{BTT89b,BCE+14a}.}

In this paper, we devise an ongoing process in which voters may arrive, leave, and change their preferences over time and 
collective decisions are made repeatedly at intervals.
Voters are never asked for their complete preference relations, but rather reveal minimal information about their preferences by choosing between two randomly drawn alternatives from time to time. No central voting authority is required. The process can be executed via a simple physical device: an urn filled with balls that allows for two primitive operations: \emph{(i)} randomly sampling a ball and \emph{(ii)} replacing a sampled ball of one kind with a ball of another kind. 
More precisely, the process works as follows (see \Cref{fig:process}). There is an urn filled with balls that each carry the label of one alternative. The initial distribution of balls in the urn is arbitrary. In each round, a randomly selected voter will draw two balls from the urn at random. Say these two balls are labeled with alternatives 1 and 2, and the voter prefers 1 to 2. She will then change the label of the second ball to 1 and return both balls to the urn.
Alternative 1 is declared the collective choice---or winner---of this round. 
After each round, with some small probability $r$ which we call \emph{mutation rate}, a randomly drawn ball is relabeled with a random alternative.

\begin{figure}[tb]
	\centering
	\parbox[t]{0.54\textwidth}{\centering\footnotesize
	\includegraphics[width=0.55\textwidth]{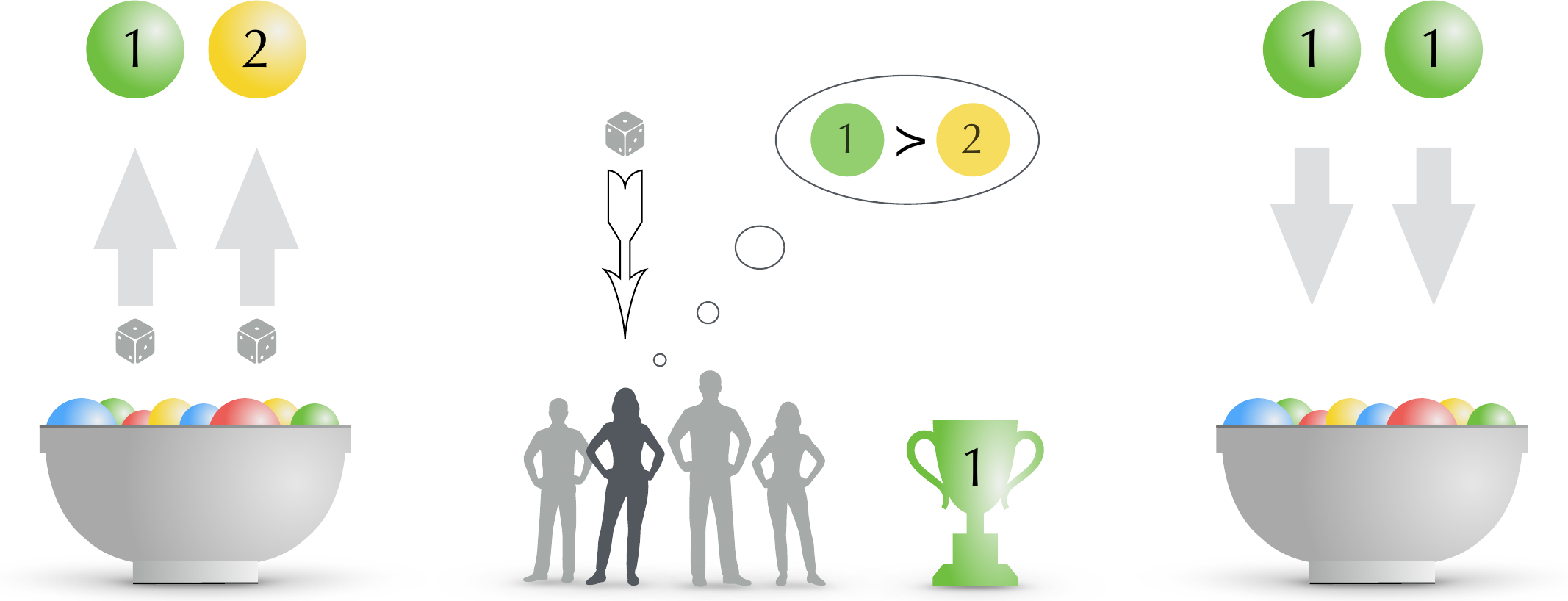}\\
	\emph{(i)} Draw two balls, random voter picks winner,\\ relabel losing ball.
	}
	\parbox[t]{0.45\textwidth}{\centering\footnotesize
	\includegraphics[width=0.55\textwidth]{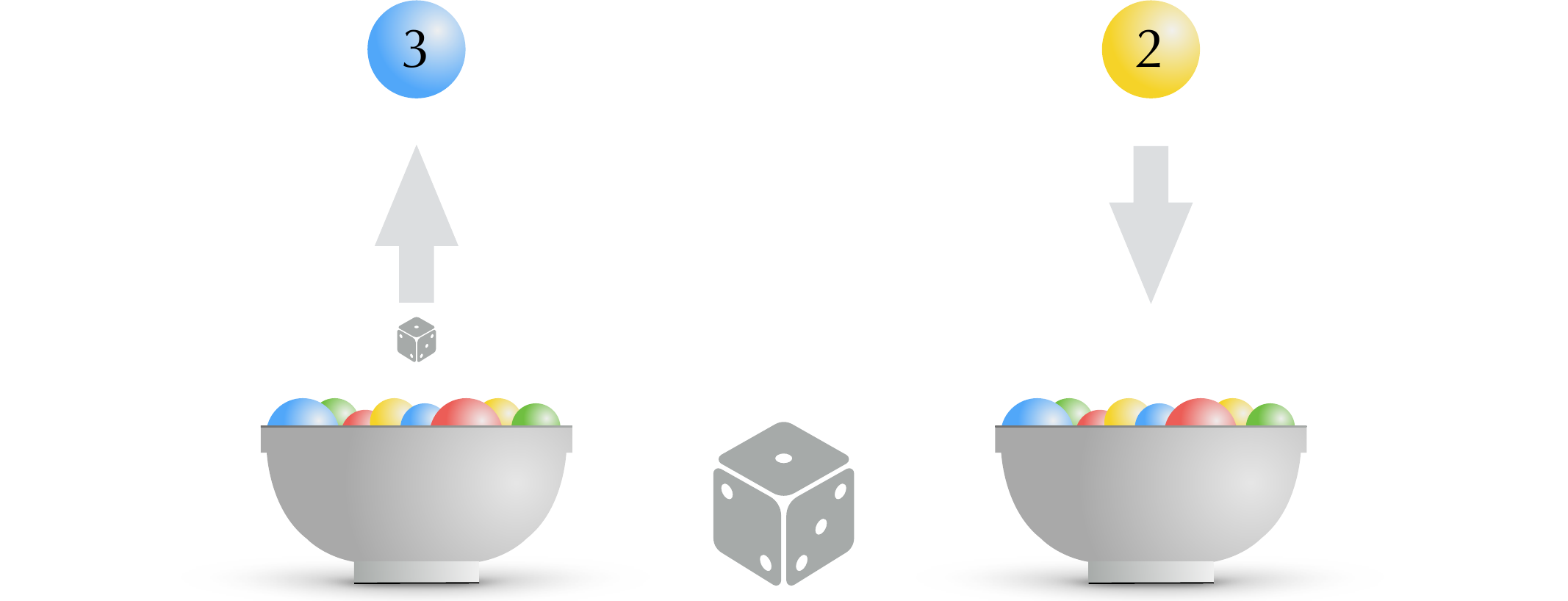}\\ 
	\emph{(ii)} Randomly relabel one drawn ball with small probability $r$ (mutation).
	}
	\\\bigskip
	\parbox[t]{0.55\textwidth}{\centering\footnotesize
	\includegraphics[width=0.55\textwidth]{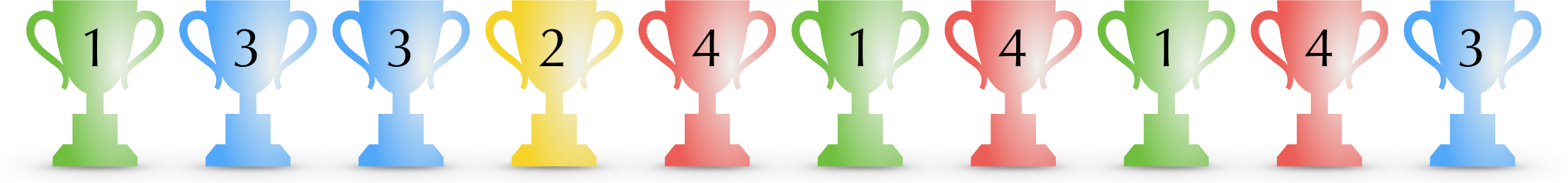}\\
	\emph{(iii)} The distribution of winners converges to a maximal lottery w.r.t.~the voters' preferences.
	}
	\caption{Illustration of one round of the urn process (\emph{i} and \emph{ii}) and the main result (\emph{iii}).}
	\label{fig:process}
\end{figure}

We show that if the number of balls in the urn is sufficiently large, then the limit of the empirical distribution of winners is almost surely close to a \emph{maximal lottery}---a randomized extension of the Condorcet principle that was proposed by \citet{Fish84a} and enjoys many desirable axiomatic properties.
How far the limiting distribution will be from a maximal lottery depends on $r$.
As $r$ goes to 0, the limiting distribution converges to a maximal lottery.
We can, however, not set $r$ to 0 as then almost surely, all alternatives except one will permanently disappear from the urn and the limiting distribution will be degenerate. Our proof not only shows convergence of the limiting distribution but also that the probability that the urn distribution itself is close to a maximal lottery gets arbitrarily close to 1 and increases exponentially in the number of rounds. The winners of most rounds are thus selected according to approximate maximal lotteries.

\subsection{Maximal Lotteries, Dynamic Voting, and Approximate Axiomatics}

The basic idea of maximal lotteries is to avoid the Condorcet paradox---which lies at the heart of classic impossibility theorems---by extending the notion of a Condorcet winner to lotteries.
A lottery $p$ is a randomized Condorcet winner---or \emph{maximal}---if for any other lottery $q$, a random voter is more likely to prefer the alternative sampled from $p$ to that sampled from $q$ than \emph{vice versa}.\footnote{This comparison of lotteries induces a binary relation on lotteries whose maximal elements are precisely the maximal lotteries.}
The minimax theorem guarantees that maximal lotteries exist. 
Maximal lotteries also have a natural interpretation in terms of electoral competition \citep[see, e.g.,][]{Myer93a,Lasl00b,CaOk07a}.
In fact, maximal lotteries are precisely the mixed Nash equilibrium (or maximin) strategies of the symmetric two-player zero-sum game given by the pairwise majority margins of the voters' preferences. When interpreting the two players as parties and the alternatives as possible positions of the parties, this can be seen as a game of electoral competition in which two parties aim at maximizing the number of voters who prefer their (mixed) position to that of the other party. For this reason, the social choice literature sometimes refers to the support of maximal lotteries as the \emph{bipartisan set} (a term proposed by Roger Myerson).

Maximal lotteries are known to satisfy a number of desirable properties that are typically considered in social choice theory \citep[see, e.g.,][]{FeMa92a,Lasl00a,RiSh10a,Hoan17a,BBS20a}. 
For example, Condorcet winners (i.e., alternatives that defeat every other alternative in a pairwise majority comparison) will be selected with probability 1, and Condorcet losers (i.e., alternatives that are defeated in all pairwise majority comparisons) will never be selected. No group of voters benefits by abstaining from an election, removing losing alternatives does not affect maximal lotteries, and each alternative's probability is unaffected by cloning other alternatives. Maximal lotteries have been axiomatically characterized using Arrow's independence of irrelevant alternatives and Pareto efficiency \citep{BrBr17a} as well as population-consistency and composition-consistency \citep{Bran13a}. 
The dynamic procedure described above implements maximal lotteries while providing 
\begin{itemize}
	\item \emph{myopic strategyproofness} within each round,
	\item minimal \emph{preference elicitation} and, thus, increased privacy protection,
	\item \emph{verifiability} realized via a simple physical procedure, and
	\item all-round \emph{flexibility}.
\end{itemize}

\noindent\emph{Myopic strategyproofness}: Each round's decision is made by letting a randomly selected voter choose between two alternatives. Clearly, a voter who is only concerned with the outcome of the current round is best off by choosing the alternative that she truly prefers. If she also takes into account the outcomes of future rounds, however, she may be able to skew the distribution in the urn by choosing alternatives strategically.\footnote{Maximal lotteries, like any \emph{ex post} Pareto efficient randomized choice procedure other than random dictatorships, fail to be strategyproof \citep{Gibb77a}. 
The simple notion of myopic strategyproofness could be strengthened by discounting future rounds.}
\medskip

\noindent\emph{Preference elicitation}: Eliciting pairwise preferences on an as-needed basis has several advantages. First, it spares the voters from the cognitive burden of having to rank-order all alternatives at once. If the number of voters is large, it may well be possible that the urn process yields satisfying results without ever querying some of the voters. Secondly, rather than submitting a complete ranking of all alternatives to a trusted authority, voters only reveal their preferences by making pairwise choices from time to time.\footnote{Privacy can be further increased by letting voters draw their balls privately, announce the winner, and put two balls of the same color back into the urn, without revealing the original color of the losing ball.
Alternatively, the voters' preferences can be protected completely by letting the voter publicly draw both balls, make one copy of each ball, and let her privately put back two balls of her choice. The collective decision in each round can then be made by drawing a random ball from the urn. 
}
\medskip

\noindent\emph{Verifiability}: Previously, the deployment of maximal lotteries required that a central authority collects the preferences of all voters, computes a maximal lottery by solving a linear program, and instantiates the lottery in some user-verifiable way. The urn process allows to achieve these goals via a simple physical device. 
\medskip

\noindent\emph{Flexibility}: The urn process is oblivious to changes in the voters' preferences, the set of voters as well as the set of alternatives. Everything that has happened up to the current round is irrelevant. Since the process converges from any initial configuration, it will keep ``walking in the right direction'' (towards a maximal lottery of the current preference profile). If the preferences change slowly in the sense that only a small fraction of voters changes their preferences from one round to the next, collective choices will thus be made according to a maximal lottery for the current preferences in most rounds. This includes the case when the distribution of preferences converges.

\medskip

We also note some disadvantages of the urn process.
The convergence of the distribution of winners to an approximate maximal lottery is an asymptotic result. 
In particular, for a finite number of rounds, there is a non-zero probability that the chosen alternative is subpar for a significant fraction of rounds, for example, because it is Pareto dominated. 
To bound this probability below an acceptable threshold, it may be necessary to run the process for an excessively large number of rounds. 
Second, ensuring that the limit distribution is sufficiently close to a maximal lottery could require an urn with a large number of balls.
We address the first concern by showing that the probability for the distribution of winners to be far from the limit distribution converges to 0 exponentially fast in the number of rounds. 
The rate of convergence is also evident in computational simulations we ran for various parameterizations of the process.
When the preference profile admits a Condorcet winner, we can give tractable bounds on the number of balls in the urn required to achieve a good approximation in the limit.
This partially mitigates the second concern since it has been observed that most real-world preference profiles admit Condorcet winners \citep[see, e.g.,][]{GeLe11a}.
 
The axiomatic characterizations of maximal lotteries not only imply that maximal lotteries satisfy desirable axioms, but also that any deviation from maximal lotteries leads to a violation of at least one of the axioms.
Hence, a process that only guarantees an approximation of a maximal lottery will not enjoy the same axiomatic properties.
However, rather than insisting on stringent axioms, one can relax them by only requiring them to hold in an approximate sense. For example, a natural notion of approximate Condorcet-consistency would require that a Condorcet winner receives probability close to 1 whenever one exists.
Since the empirical distribution of winners according to our process is almost surely close to a maximal lottery and maximal lotteries are Condorcet-consistent, the process is approximately Condorcet-consistent in the above sense.
More generally, approximate maximal lotteries satisfy approximate versions of many of the axioms enjoyed by maximal lotteries such as population-consistency, composition-consistency, agenda-consistency, and efficiency.
This follows from the fact that the correspondence returning the set of maximal lotteries depends continuously on the underlying preference profile and we show this exemplarily for population-consistency in \Cref{sec:approximateaxiomatics}.

Maximal lotteries have been repeatedly recommended for practical use \citep{FeMa92a,RiSh10a,Bran13a,Hoan17a}. 
We believe that the benefits of the urn process described above extend the applicability of maximal lotteries. 
Rather than for traditional political elections, probabilistic rules like maximal lotteries seem more suitable for 
frequently repeated low-stakes elections where some degree of randomization may not only be tolerable but even desirable.
Two example applications that have been suggested for maximal lotteries are to help a group of coworkers with the daily decision where to have lunch and to select music for a party or a radio station based on the preferences of the listeners \citep{Bran13a}. The transparency and the flexibility of the urn process seem particularly effective in the music broadcasting example. Agents come and go, they only need to select from a pair of songs rather than rank-order all of them, and individual preferences, as well as the set of available songs, can be changed at any time. Our theorem shows that the sequence of simple pairwise choices results in a socially desirable distribution of songs:
the more songs are being played, the less likely it becomes that another distribution of songs would have been preferred by an expected majority of listeners. 
It is plausible that, over time, the preferences of the listeners change depending on the songs that have been played so far. These changes will be reflected immediately in the selection of future songs.

\subsection{Applications Beyond Collective Decision-Making}

Interestingly, dynamic processes similar to the process we describe here have recently been studied in population biology, quantum physics, chemical kinetics, and plasma physics to model phenomena such as the coexistence of species, the condensation of bosons, the reactions of molecules, and the scattering of plasmons. In each of these cases, simple interactions between randomly sampled entities result in distributions that correspond to equilibrium strategies of symmetric zero-sum games. Since the definition of maximal lotteries and our dynamic process merely rely on this comparison matrix, describing with which probability one entity will be replaced with another in a pairwise encounter, our results are also of relevance to these areas. We discuss these connections, as well as those to equilibrium learning and evolutionary game theory, in more detail in \Cref{sec:relatedwork}.

An alternative interpretation of our result can be used to describe the formation of opinions. In this model, there is a population of agents, each of whom entertains one of many possible opinions. Agents come together in random pairwise interactions, in which they try to convince each other of their opinion. 
The probabilities with which one opinion beats another are given as a square matrix and, with some small probability, an agent randomly changes her opinion. In other words, the agents correspond to the balls in the urn, the opinions correspond to the alternatives, and there are neither voters nor preference profiles as transition probabilities are given explicitly.
Our theorem then shows that, if the population is large enough, the distribution of opinions within the population is close to a maximal lottery of the probability matrix most of the time.
Other models of opinion formation based on different processes were, for example, considered by \citet{Degr74a}, \citet{HoLi75a}, and \citet{GoLe14a}.

The process we describe approximately computes a mixed Nash equilibrium of a symmetric zero-sum game. This problem is known to be equivalent to linear programming. In fact, deciding whether an action is played with positive probability in an equilibrium of a symmetric zero-sum game is P-complete
\citep[][Theorem 5]{BrFi08b}, which, loosely speaking, means that the problem is at least as hard as any problem that can be solved in polynomial time. The urn process can thus be seen as a probabilistic algorithm that approximates polynomial-time computable functions. In contrast to traditional computing devices such as Turing machines, the urn process is based on unordered elementary entities that randomly interact according to very simple replacement rules.\footnote{Related decentralized models of computation with applications to sensor networks and molecular computing are studied under the name ``population protocols'' in computer science \citep[e.g.,][]{AAD+06a,AsRu09a}. While the urn process has the same \emph{modus operandi} as population protocols, the input-output behavior is different. The input of population protocols is given by the initial distribution of balls in the urn and the output has been reached if all balls belong to a certain subset of types. By contrast, the input for our urn process is encoded in the matrix describing the replacement rules and the (approximate) output is given by the distribution of balls in the urn after sufficiently many rounds.}

\medskip

The remainder of the paper is structured as follows. After defining our model in \Cref{sec:model}, we state the main result (\Cref{thm:main}) and a rough proof sketch in \Cref{sec:results}. The full proof is given in the Appendix. In \Cref{sec:cond}, we analyze the instructive special case of preference profiles that admit a Condorcet winner, which allows for a more elementary proof. 
In \Cref{sec:relatedwork}, we extensively discuss the connections between our work and results from equilibrium learning, evolutionary game theory, and the natural sciences. We also state a continuous version of our main result (\Cref{thm:continuous}) that may be of independent interest.
In \Cref{sec:approximateaxiomatics}, we show in which sense the axiomatic properties of maximal lotteries can be retained for approximations thereof.

\section{The Model}\label{sec:model}

Let $[d]=\{1,\dots,d\}$ be a set of alternatives and $\Delta$ the $d-1$-dimensional unit simplex in $\mathbb R^d$, that is, $\Delta = \{x \in\mathbb R_{\ge 0}^d\colon \sum_{i=1}^d x_i = 1\}$.
We refer to elements of $\Delta$ as lotteries.
By $\mathbb N = \{1,2,\dots\}$ and $\mathbb N_0 = \mathbb N\cup\{0\}$ we denote the sets of positive and non-negative integers, respectively.
Throughout the paper, for a vector $x\in\mathbb R^k$ for some $k$, $\lvert x\rvert = \sum_{l = 1}^k \lvert x_l\rvert$ denotes its $L^1$-norm.
For $\delta > 0$ and $S\subset\mathbb R^d$, let $B_\delta(S) = \{x\in\Delta\colon \lvert x - y\rvert < \delta \text{ for some } y\in S\}$ be the $\delta$-ball around $S$.
For a finite set $S$, we write $\lvert S\rvert$ for the number of elements of $S$.

A \emph{preference relation} $\s$ is an asymmetric binary relation over $[d]$.\footnote{Preferences need not be transitive or complete. The definition of maximal lotteries and the urn process we describe only depend on the fractions of voters who prefer one alternative to another. In particular, indifferences can easily be accommodated by randomly selecting which of the two balls will be relabelled.} 
By $\mathcal R$ we denote the set of all preference relations.
Let $V$ be a finite set of voters.
A \emph{preference profile} $\p\in\mathcal R^V$ specifies a preference relation for each voter.
With each preference profile $\p$, we can associate a comparison matrix $M_{\p}\in[0,1]^{d\times d}$ that states for each ordered pair of alternatives the fraction of voters who prefer the first to the second. 
That is, $M_{\p}(i,j) =  \lvert\{v\in V\colon i \srel[v] j\}\rvert / \lvert V\rvert$.
This matrix induces a skew-symmetric matrix $\tilde M_{\p} = M_{\p} - M_{\p}^\intercal$, which we call the skew-comparison matrix.\footnote{A matrix $M$ is skew-symmetric if $M = -M^\intercal$.}

\subsection{Maximal Lotteries}\label{sec:ml}

A lottery $p\in\Delta$ is a \emph{maximal lottery} for a profile $\p$ if $\tilde M_{\p}\,p \le 0$. The minimax theorem implies that every profile admits at least one maximal lottery.
By $\ml(\p)$ we denote the set of all lotteries that are maximal for $\p$. Most profiles admit a unique maximal lottery. For example, when the number of voters is odd and voters have strict preferences, there is always a unique maximal lottery \citep{LLL97a}. 

\begin{example}[Condorcet winner]\label{ex:condorcetwinner}
Consider, for example, 900 voters, three alternatives, and a preference profile $\p$ given by the following table. Each column header contains the number of voters with the corresponding preference ranking.
\[
\begin{tabular}{ccc}
$300$ & $300$ & $300$\\
\midrule
$1$ & $1$ & $2$\\
$2$ & $3$ & $3$\\
$3$ & $2$ & $1$\\
\end{tabular}
\]
Then,
\[
M_{\p} =
\begin{pmatrix}
0 & \nicefrac23 & \nicefrac23 \\
\nicefrac13 & 0 & \nicefrac23 \\
\nicefrac13 & \nicefrac13 & 0\\
\end{pmatrix}
\quad\text{and}\quad
\tilde M_{\p} =
\begin{pmatrix}
\phantom{-}0 & \phantom{-}\nicefrac13 & \nicefrac13\\
-\nicefrac13 & \phantom{-}0 & \nicefrac13\\
-\nicefrac13 & -\nicefrac13 & 0\\
\end{pmatrix}\text.
\]
The set of maximal lotteries $\ml(\p)=\{ (1,0,0)^\intercal \}$ only contains the degenerate lottery with probability 1 on the first alternative. This alternative is a \emph{Condorcet winner}, i.e., an alternative that is preferred to every other alternative by some majority of voters.
\end{example}

\subsection{Markov Chains}\label{sec:markovchains}

Let $S$ be a finite set and $\{X(n)\colon n \in\mathbb N_0\}$ be a discrete-time, time-homogeneous \emph{Markov chain} with state space $S$.
The \emph{transition probability matrix} $P\in[0,1]^{S\times S}$ of $\{X(n)\colon n \in\mathbb N_0\}$ is given by
\begin{align*}
	P(p,p') = \prob{X(n+1) = p'\mid X(n) = p}
\end{align*}
for all $p,p' \in S$.
We will frequently write $X(n,p_0)$ for $X(n)$ conditioned on $X(0) = p_0\in S$ and call $p_0$ the \emph{initial state}.  

The period of a state $p\in S$ is the greatest common divisor of the return times with positive probability $\{n\in\mathbb N\colon (P^n)(p,p) > 0\}$. 
A Markov chain is \emph{aperiodic} if every state has period 1. 
Note that any Markov chain with $P(p,p) > 0$ for all $p\in S$ is aperiodic.
A Markov chain is \emph{irreducible} if every state is reached from any other state with positive probability.
That is, for any two states $p,p'\in S$, there is a positive integer $n$ so that $(P^n)(p,p') > 0$. 
If $\{X(n)\colon n \in \mathbb N_0\}$ is irreducible and aperiodic, it has a unique \emph{stationary distribution} $\pi\in\Delta S$ so that $\pi^\intercal = \pi^\intercal P$.

\subsection{The Urn Process}\label{sec:urnprocess}

Consider an urn with $N\in\mathbb N$ balls, each labeled with some alternative.
Viewing balls with the same label as indistinguishable, we can identify each state of the urn with an element of the discrete unit simplex $\Delta^{(N)} = \{p\in\Delta\colon Np\in\mathbb N_0^d\}$.
Fix a \emph{mutation rate} $r\in[0,1]$.

We are interested in a Markov chain with state space $\Delta^{(N)}$ that can be informally described as follows.
First, we flip a coin that has probability $1-r$ of landing heads. 
If the coin shows heads, we choose one voter $v\in V$ uniformly at random and ask the voter to draw two balls from the urn. Say these two balls are labeled with alternatives 1 and 2. If $1 \succ_v 2$, the label of the second ball is changed to 1. Likewise, if $2 \succ_v 1$, the first ball is relabeled with label 2. If both balls carry the same label, the labels remain unchanged. 
If the coin shows tails, we draw a single ball from the urn, relabel it with an alternative chosen uniformly at random, and put it back into the urn.

This description of the process assumes that two alternatives are sampled from the urn distribution \emph{without} replacement.
For the formal description, we will assume that drawing is \emph{with} replacement.
This corresponds to sampling one alternative by drawing one ball, putting the ball back into the urn, and sampling a second alternative by again drawing one ball (which may be the same as the first).\footnote{In practice, it would the be infeasible to relabel the first drawn ball since it is ``lost'' in the urn after putting it back. However, one could modify the process by relabeling the second ball if the voter prefers the first sampled alternative and change nothing otherwise. This eliminates the factor of 2 in the transition probabilities below but does not change the results.}
Doing so avoids a lot of clumsy notation.
If the number of balls in the urn is large, there is no significant difference between drawing with and without replacement.
In the proof, we point out why the same arguments also carry through with minor adaptations for drawing without replacement.  

We define a transition probability matrix $P^{(N,r)}$ that specifies for every pair of states the probability that the distribution of the urn transitions from the first to the second.
Denote by $e_i$ the $i$th unit vector in $\mathbb N_0^d$.
For $p\in \Delta^{(N)}$ and $i,j\in[d]$ with $p' = p + \frac{e_i}N - \frac{e_j}N\in \Delta^{(N)}$, let 
\begin{align*}
	P^{(N,r)}(p,p') =
	\begin{cases}
		\displaystyle (1-r)2p_ip_jM_{\p}(i,j) + \frac rd p_j\qquad&\text{if }i \neq j\\
		\displaystyle (1-r) \sum_{k = 1}^d p_k^2 + \frac rd &\text{if } i = j\\
	\end{cases}
\end{align*}
be the probability of transitioning from $p$ to $p'$.
For the remaining pairs of states $p,p'\in \Delta^{(N)}$, let $P^{(N,r)}(p,p') = 0$.
Then, $P^{(N,r)}$ has non-negative values and its rows sum to 1 so that it is a valid transition probability matrix.
For an initial state $p_0\in \Delta^{(N)}$, we consider a Markov chain $\{X^{(N,r)}(n,p_0)\colon n\in\mathbb N_0\}$ with transition probability matrix $P^{(N,r)}$.
The distribution of $X^{(N,r)}(n,p_0)$ over $\Delta^{(N)}$ is given by the row of $\left(P^{(N,r)}\right)^n$ with index $p_0$.
If $r > 0$, this Markov chain is irreducible and aperiodic (since it remains in the same state with positive probability).
It corresponds to the urn process described above when the initial state of the urn is $p_0$.

Continuing \Cref{ex:condorcetwinner}, consider an urn with $N=5$ balls and recall that $d = 3$. Then, the transition probability matrix $P^{(N,r)}$ is an $\binom{3+5-1}{5}=21$-dimensional square matrix. Let the mutation rate be $r=0.1$ and the initial state $p_0=\frac15\,(1,2,2)^\intercal$. 
The probability that one of the balls of the second type is replaced with one of the first type is 
\[P^{(5,0.1)}(p_0,\frac15\,(2,1,2)^\intercal)=0.9 \cdot \frac{4}{25} \cdot \frac{2}{3}+0.1\cdot \frac13\cdot \frac{2}{5} \sim 0.109\text.\]

\section{The Result}\label{sec:results}

We prove the following: 
\begin{quote}
	\emph{For any small enough mutation rate $r > 0$, there is a maximal lottery $p^*$ so that for any initial state $p_0$, $X^{(N,r)}(k,p_0)$ is close to $p^*$ for all but a small fraction of rounds $k$ provided that the number of balls $N$ is large enough.}
\end{quote}
More precisely, for any $\delta,\tau > 0$, there is an upper bound on the mutation rate $r_0 > 0$ so that for every $0 < r \le r_0$, there is a maximal lottery $p^*$ and a lower bound on the number of balls $N_0\in\mathbb N$ such that for every $N \ge N_0$ and every $p_0\in \Delta^{(N)}$, the fraction of rounds $k$ in which $X^{(N,r)}(k,p_0)$ is no more than $\delta$ away from $p^*$ is almost surely at least $1-\tau$.\footnote{\Cref{thm:main} implies that the stationary distribution of $X^{(N,r)}$ assigns probability at least $1-\tau$ to states that are in a $\delta$-neighborhood of $p^*$.
Conversely, this property of the stationary distribution implies \Cref{thm:main} by the ergodic theorem for Markov chains.
The proof does however not derive the above property of the stationary distribution as an intermediate step. 
It is only a by-product of the final result.
For more than two alternatives, our result is stronger than proving that the expectation of the stationary distribution, or, equivalently, the temporal average of the urn distribution, is close to a maximal lottery.}

\begin{restatable}{theorem}{main}\label{thm:main}
	Let $\delta,\tau > 0$.
	Then, there is $r_0>0$ such that for all $0< r\le r_0$, there are $p^*\in\ml(\p)$ and $N_0\in\mathbb N$ such that for all $N \ge N_0$ and $p_0\in \Delta^{(N)}$, almost surely
	\begin{align*}
		 \lim_{n\rightarrow\infty} \frac1n \left\lvert\left\{k \le n\colon \left\lvert X^{(N,r)}(k,p_0) - p^*\right\rvert \le \delta\right\}\right\rvert \ge 1 - \tau.
	\end{align*}
	Moreover, there is $C>0$ such that for all $n\in \mathbb N_0$,
	\[
	\mathbb P\left(\left\lvert X^{(N,r)}(n,p_0) - p^*\right\rvert \le \delta\right) \ge 1-\tau -e^{-\lfloor Cn\rfloor}\text.
	\]
\end{restatable}

To prove this, we approximate our discrete and stochastic urn process by a continuous and deterministic process. 
The latter can be viewed as a version of the urn process with a continuum of balls.
Using analytical tools, it can be shown that this process converges to an approximate maximal lottery for every initial state, where the approximation can be made arbitrarily precise if $r$ is made small (see \Cref{thm:continuous} in \Cref{sec:relatedwork}).
The approximation only works for a finite number of rounds (respectively, bounded time interval) and only with probability close to 1 (rather than almost surely).
However, on long enough time intervals, the deterministic process is close to an approximate maximal lottery most of the time (since it converges to such a lottery). 
Moreover, on any such interval, the deterministic process is a good approximation to the stochastic process with probability close to 1 (provided that the number of balls is large enough).
By a variant of the strong law of large numbers, it then follows that the stochastic process is close to an approximate maximal lottery for most rounds almost surely.
This is the first statement of \Cref{thm:main}.
We give a more detailed outline and a complete proof in the Appendix.
The second statement follows from the first using the standard result that the distribution of an irreducible and aperiodic Markov chain converges exponentially fast to its stationary distribution in the total variation norm.

\Cref{thm:main} is a statement about the distribution in the urn. 
Recall that the collective decision in each round is the winner of the pairwise comparison between the two drawn balls.
It is not hard to show that the empirical distribution of winners is also close to a maximal lottery.\footnote{Suppose the distribution of balls in the urn is $p\in\Delta^{(N)}$.
Then the probability that $i\in[d]$ is the collective decision is
\begin{align*}
	w_i = p_i \Big(p_i + 2\sum_{j\neq i} M_{\p}(i,j) p_j\Big)
	= p_i \Big(p_i +  \sum_{j\neq i} (\tilde M_{\p}(i,j) + 1)p_j\Big) = p_i \Big(1 + \tilde M_{\p}p\Big)
\end{align*}
where we used that $2M_{\p}(i,j) = \tilde M_{\p}(i,j) + 1$, $\sum_{j\in[d]} p_j = 1$, and $M_{\p}(i,i) = 0$.
If $p^*$ is a maximal lottery and $\lvert p - p^*\rvert\le \delta$, then $(\tilde M_{\p}p)_i \le \delta$ for all $i\in[d]$.
Hence, $w_i \in [p_i - \delta, p_i + \delta]$ for all $i$, so that $\lvert w - p^*\rvert \le (d+1)\delta$.
For every $\delta' > 0$, choosing $\delta = \tau = \frac{\delta'}{2(d+1)}$ in \Cref{thm:main} thus shows that the empirical distribution of collective decisions is almost surely no more than $\delta'$ away from $p^*$.}

Another straightforward corollary of \Cref{thm:main} is that the temporal average of the urn distribution is almost surely close to a maximal lottery (provided that $r$ is small and $N$ is large).
Let
\begin{align*}
	Z^{(N,r)}(n,p_0) = \frac1{n}\cdot\sum_{k = 0}^{n-1} X^{(N,r)}(k,p_0)
\end{align*}
be the temporal average of $X^{(N,r)}(k,p_0)$ over the first $n$ rounds.
Then, we have the following.

\begin{restatable}{corollary}{cormain}\label{cor:main}
	Let $\delta > 0$.
	Then, there is $r_0 > 0$ such that for all $0 < r \le r_0$, there are $p^*\in\ml(\p)$ and $N_0\in\mathbb N$ such that for all $N\ge N_0$ and $p_0\in \Delta^{(N)}$,
	\begin{align*}
			\prob{\left\lvert\lim_{n\rightarrow\infty}Z^{(N,r)}(n,p_0) - p^*\right\rvert \le \delta} = 1
	\end{align*}
\end{restatable}
\begin{proof}
	For some $\tau$ to be determined later, let $r_0$ and, depending on $0< r \le r_0$, $p_*$ and $N_0$ be as obtained from \Cref{thm:main}.
	By the triangle inequality, we have
	\begin{align*}
		\left\lvert Z^{(N,r)}(n,p_0) - p^*\right\rvert \le \frac1n\cdot\sum_{k = 0}^{n-1} \left\lvert X^{(N,r)}(k,p_0) - p^*\right\rvert.
	\end{align*}
	By \Cref{thm:main}, in the limit when $n$ goes to infinity, all but a $1-\tau$ fraction of the summands on the right-hand side are smaller than $\delta$.
	The remaining summands are bounded by 2.
	Hence, choosing $\tau = \frac\delta2$ gives that almost surely, $\lvert\lim_{n\rightarrow\infty} Z^{(N,r)}(n,p_0) - p^*\rvert \le 2\delta$.
	The ergodic theorem for Markov chains ensures that the limit exists.
\end{proof}

Before illustrating these results via examples, we discuss variations of the urn process.

\begin{remark}[Decoupling collective decisions]
	We assume that in each round, a collective decision is made by selecting the winner of the pairwise comparison.
	It may, however, be more practical to decouple collective decisions from the preference elicitation process and draw winners less frequently.
	For example, collective decisions could be made by drawing a random ball after any fixed number of rounds or at random times.
	\Cref{cor:main} shows that the resulting distribution would approximate a maximal lottery.	
\end{remark}

\begin{remark}[Non-uniform mutation rates]
The results still hold if we let the probability of a random mutation from alternative $i$ to alternative $j$ depend on the pair $(i,j)$. 
It suffices that every alternative in the support of a maximal lottery can escape permanent depletion via some path of mutations.
More explicitly, it suffices if for any two alternatives $i$ and $j$, there is a path of alternatives from $i$ to $j$ so that the mutation rate is positive from any alternative on the path to the next.
The proof can be adapted at the expense of more book-keeping.
\end{remark}

\begin{remark}[Mutation rate vs.\ urn size]\label{rem:smallmutationrate}
	\Cref{cor:main} shows that the temporal average of the urn distribution converges to a maximal lottery if we let $N$ go to infinity and then take $r$ to 0 (see also \Cref{thm:continuous}).
	This is in contrast to other works on evolutionary dynamics that take limits in the reverse order.\footnote{For example, \citet{FuIm08a} study imitation dynamics with mutations in symmetric two-player games (not necessarily zero-sum).
	They consider the case when the mutation rate goes to 0 for a fixed population size $N$.
	For small but positive mutation rates, the dynamics spend most of the time in degenerate states where all but a small fraction of individuals play the same strategy.  
	Letting the mutation rate go to 0 thus induces a distribution over actions. 
	Their main result determines the limit of this distribution as the population size goes to infinity.}
	While the frameworks are similar, these results are conceptually different. 
	In our model, if $r$ is too small compared to $\frac1N$, it will, in general, not be the case that the distribution in the urn is close to a maximal lottery for most rounds.
	For any long enough time interval, the distribution in the urn will for all $r$ degenerate within the interval with high probability, that is, it will only contain balls of one type.
	If $r$ is very small, it will stay in a degenerate state for a long time (compared to the chosen interval) with high probability.
	When the process leaves the degenerate state, the same will repeat itself (possibly with a different degenerate state), so that the process spends most rounds in degenerate states. As a consequence, decreasing $r$ over time does not work unless $N$ is increased as well. When increasing $N$ at an appropriate rate, the urn distribution will converge exactly to the set of maximal lotteries by \Cref{thm:main}.
\end{remark}

\begin{remark}[Majority voting]
Rather than letting only a single voter decide on the pairwise comparison between the two randomly drawn balls, it is possible to ask all voters which alternative they prefer and replace the alternative which is less preferred by a \emph{majority} of voters. This variant is equivalent to the original process for a single voter with possibly intransitive preferences (given by the majority relation of the entire population of voters) and converges to a so-called C1 maximal lottery of the preference profile \citep[see][for more information on C1 maximal lotteries]{BBS20a}.
\end{remark}

\begin{remark}[Static or growing urn]
When the initial distribution of balls in the urn is uniform and remains fixed (i.e., no balls are replaced over time), then the empirical distribution of winners converges to the lottery returned by the \emph{proportional Borda rule} \citep[see, e.g.,][]{Barb79b,Heck03a}.\footnote{The proportional Borda rule assigns to each alternative a probability that is proportional to its Borda score. For example, for one voter with lexicographic preferences over $a$, $b$, and $c$, the Borda scores are $2$, $1$, and $0$, respectively. The proportional Borda rule thus returns the lottery $(\nicefrac23, \nicefrac13,0)$.} 
This rule violates Condorcet-consistency and Pareto efficiency. It can put probability $\frac{1}{d}$ on Pareto-dominated alternatives and almost as little as $\frac{1}{d}$ on Condorcet winners for large numbers of voters \citep{BLR21b}.
When adding a new ball labeled with the winning alternative rather than replacing the losing one (i.e., the number of balls increases over time), neither the relative distribution in the urn nor the temporal average converges (see \Cref{sec:relatedwork}).
\end{remark}

\Cref{fig:ternary} (left) shows a simulation of the urn process for the preference profile and corresponding skew-comparison matrix given in \Cref{ex:condorcetwinner}. The urn process corresponds to a random walk within the shown triangle starting from the center (an almost uniform distribution). The first alternative in this profile is a Condorcet winner. From round 177 on, at least 90\% of the balls (45 of the 50) are labeled with the Condorcet winner except for three rounds. At this point, only 160 of the 900 voters were asked for their preferences.
The path is tilted to the left because a majority of voters prefer alternative 2 to alternative 3.
Note that the process only depends on the fractions of voters who prefer one alternative to another and is, thus, independent of the number of voters. Hence, if there are nine million---rather than nine hundred---voters whose preferences are distributed as in \Cref{ex:condorcetwinner}, the process could turn out exactly as shown in \Cref{fig:ternary}. In particular, the overwhelming majority of voters would never be queried for their preferences. 

\begin{figure}[tb]
	\centering
	\includegraphics{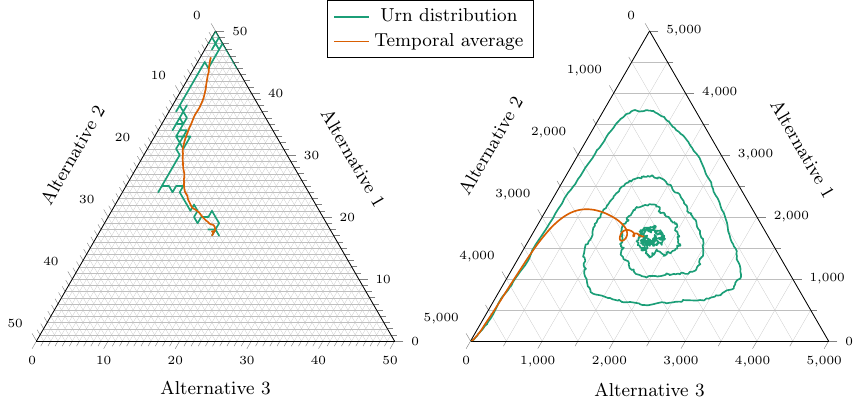}
	\caption{Simulations of the urn process.\\
	 The left diagram shows the urn process for the profile given in \Cref{ex:condorcetwinner} using an urn with $N = 50$ balls for 1,000 rounds and mutation rate $r = 0.02$, starting from an almost uniform distribution. Each intersection of the grid lines corresponds to a configuration of the urn. 
	The right diagram shows the urn process for the profile given in \Cref{ex:condorcetcycle} using an urn with $N = 5,000$ balls for 500,000 rounds and mutation rate $r = 0.04$, starting from the degenerate distribution in which all balls are labeled with Alternative 2. 
	The green lines depict the actual distribution of balls while the red lines depict the temporal average of urn distributions until the given round.
	}
	\label{fig:ternary}
\end{figure}

We now give two other examples, for which the unique maximal lottery is not degenerate.

\begin{example}[Condorcet cycle]\label{ex:condorcetcycle}
	Consider 900 voters, three alternatives, and the following preference profile $\p$, leading to a so-called Condorcet cycle or Condorcet paradox.
	\[
	\begin{tabular}{ccc}
	$300$ & $300$ & $300$\\
	\midrule
	$1$ & $2$ & $3$\\
	$2$ & $3$ & $1$\\
	$3$ & $1$ & $2$\\
	\end{tabular}
	\]
	Then,
	\[
	M_{\p} =
	\begin{pmatrix}
	0 & \nicefrac23 & \nicefrac13 \\
	\nicefrac13 & 0 & \nicefrac23 \\
	\nicefrac23 & \nicefrac13 & 0\\
	\end{pmatrix}
	\quad\text{and}\quad
	\tilde M_{\p} =
	\begin{pmatrix}
	\phantom{-}0 & \phantom{-}\nicefrac13 & -\nicefrac13\\
	-\nicefrac13 & \phantom{-}0 & \phantom{-}\nicefrac13\\
	\phantom{-}\nicefrac13 & -\nicefrac13 & \phantom{-}0\\
	\end{pmatrix}\text.
	\]
	The set of maximal lotteries $\ml(\p) = \{(\nicefrac13,\nicefrac13,\nicefrac13)\}$ consists of the uniform lottery over the three alternatives.
	A simulation of the urn process for this profile is given in \Cref{fig:ternary} (right). This time, the initial distribution is degenerate with all balls being of type 2. It can be seen how the distribution of balls in the urn closes in on the maximal lottery and remains in its neighborhood for most of the time while the temporal average converges to the maximal lottery.
\end{example}

\begin{example}\label{ex:3cycle+condorcetloser}
Consider the following preference profile $\p$ with 900 voters and four alternatives.
	\[
	\begin{tabular}{ccc}
	$375$ & $300$ & $225$ \\
	\midrule
	$1$ & $3$ & $4$\\
	$2$ & $1$ & $2$\\
	$3$ & $2$ & $3$\\
	$4$ & $4$ & $1$\\
	\end{tabular}
	\]
Then,
	\[
	\tilde M_{\p} =
	\begin{pmatrix}
	\phantom{-}0 & \phantom{-}\nicefrac13 & -\nicefrac19 & \nicefrac13\\
	-\nicefrac13 & \phantom{-}0 & \phantom{-}\nicefrac29 & \nicefrac13\\
	\phantom{-}\nicefrac19 & -\nicefrac29 & \phantom{-}0 & \nicefrac13\\
	-\nicefrac13 & -\nicefrac13 & -\nicefrac13 & \phantom{-}0\\
	\end{pmatrix}\text.
	\]
	The set of maximal lotteries $\ml(\p) = \{(\nicefrac13,\nicefrac16,\nicefrac12,0)\}$ consists of a single lottery, which is supported on the first three alternatives.
	A simulation of an urn process for this profile starting from the uniform distribution is given in \Cref{fig:3cycle+condorcetloser}. The figure shows the distribution in the urn, the temporal average of urn distributions, and the difference of the urn distribution and the maximal lottery in terms of the relative entropy.
\end{example}

\begin{figure}[tb]
	\centering
	\includegraphics{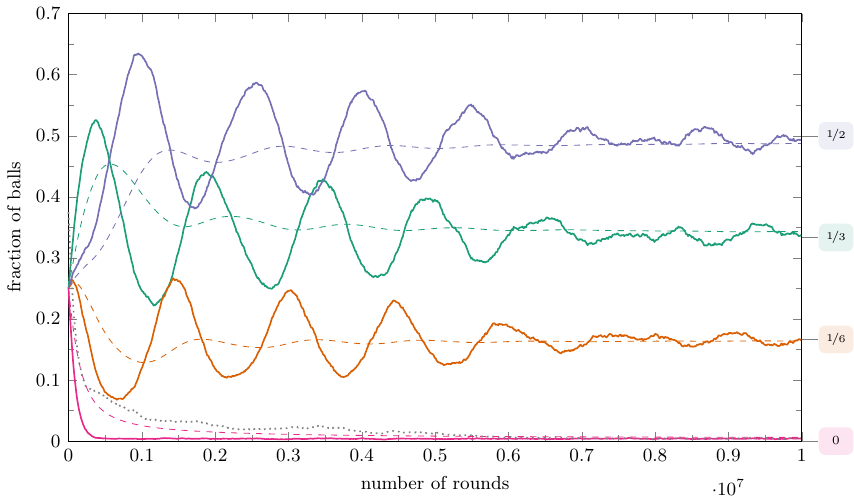}
	\caption{Simulation of the urn process for the profile in \Cref{ex:3cycle+condorcetloser} on an urn with $N = 50,000$ balls for $10^7$ rounds and mutation rate $r = 0.01$. The solid lines show the fraction of balls in the urn. The dashed lines show the temporal average of the fraction of balls in the urn until the given round. The unique maximal lottery is $p=(\nicefrac13,\nicefrac16,\nicefrac12,0)$. The dotted line shows the relative entropy $D(p\mid q) = \sum_{i\in[d]} p_i \log(\frac{p_i}{q_i})$ of $p$ with respect to the distribution in the urn~$q$.}
	\label{fig:3cycle+condorcetloser}
\end{figure}

We use the relative entropy (rather than, for example, the distance $|p-q|$) to measure how much the distribution in the urn diverges from the maximal lottery since the proof of \Cref{thm:main} shows that the entropy of the maximal lottery relative to the continuous approximation of the discrete process converges monotonically to~0.

\section{The Case of a Condorcet Winner}\label{sec:cond}

We give an elementary proof of \Cref{thm:main} for profiles that admit a Condorcet winner.
For those profiles, the unique maximal lottery assigns probability 1 to the Condorcet winner.
To analyze the stationary distribution $\pi\in\Delta(\Delta^{(N)})$ of the Markov chain induced by the urn process, it suffices to examine the fraction of balls labeled with the Condorcet winner.
This allows us to relate the Markov chain to a process that is one-dimensional in the sense that each state can only transition to two different states, and is, thus, easy to analyze.
It also enables us to give a concrete lower bound on the number of balls $N$ required for given $\delta,\tau>0$ for the conclusion of \Cref{thm:main} to hold.

Let $M = M_{\p}$ be the majority matrix of a profile $\p$ with Condorcet winner $i\in[d]$.
Hence, $M_{ij} > \frac12$ for all $j\in [d] \setminus \{i\}$.
Let $\alpha = \min\{M_{ij}\colon j\in [d] \setminus \{i\}\} - \frac 12$.
We slice up $\Delta^{(N)}$ into the level sets of the map $p\mapsto p_i$.
For $k\in\{0,\dots,N\}$, let $S_k = \{p \in\Delta^{(N)}\colon p_i = \frac{k}N\}$ be the states corresponding to distributions with $k$ of the $N$ balls of type $i$.
Then $\sigma_k := \sum_{p\in S_k} \pi(p)$ is the limit probability that the urn is in a state in $S_k$ as the number of rounds goes to infinity. 
We want to show that if $r$ is sufficiently small and $N$ sufficiently large, $\pi$ has most of the probability on states in $S_k$ with $k$ close to $N$.

For 4 alternatives, one can illustrate the ensuing argument as follows.
The set of states $\Delta^{(N)}$ corresponds to rooms in a tetrahedral-shaped pyramid. 
The rooms on the $k$th floor correspond to $S_k$, so that the tip of the pyramid is the state where all balls are of type $i$.
The urn process is a random walk through the pyramid, moving from one room to an adjacent one (which could be on the same floor, the floor below, or the floor above).
With the exception of few floors close to the tip, the probability of going up is always larger than the probability of going down.
It is then intuitively clear that if the pyramid is large enough, one should expect to find the random walk close to the tip of the pyramid most of the time.\footnote{In the analysis of the general case, the number of balls of type $i$ is replaced by the entropy of the urn distribution relative to a maximal lottery. The fact that the number of balls not of type $i$ more likely than not decreases corresponds to the fact that the expected entropy relative to a maximal lottery decreases.}

Recall that $P^{(N,r)}(p,q)$ is the probability of transitioning from state $p$ to state $q$.
Since $\pi$ is a stationary distribution, we have $\pi^\intercal P^{(N,r)} = \pi^\intercal$.
Consider any partition of $\Delta^{(N)}$ into two sets.
For the stationary distribution, the probability of transitioning from the first set to the second is equal to the probability of transitioning from the second set to the first since the probabilities of both sets are conserved.
Applying this to the sets $\bigcup_{l = 0}^{k-1} S_l$ and $\bigcup_{l = k}^N S_l$ for $k\in[N]$ and noticing that the only transitions between the two sets with positive probability are from $S^{k-1}$ to $S^{k}$ and vice versa,
we get 
\begin{align}
	\sum_{p\in S_{k-1}} \pi(p) \sum_{q\in S_k} P^{(N,r)}(p,q) = \sum_{p\in S_{k}} \pi(p) \sum_{q\in S_{k-1}} P^{(N,r)}(p,q)\text.\label{eq:updown}
\end{align}
That is, the probability of being in a state in $S_{k-1}$ and transitioning to a state in $S_k$ equals the probability of being in a state in $S_k$ and transitioning to a state in $S_{k-1}$.

Now observe that for $p \in S_k$, $k \in \{0,\dots,N-1\}$, we have
\begin{align*}
	\sum_{q\in S_{k+1}} P^{(N,r)}(p,q) &\ge 2(1-r)\frac{k(N-k)}{N^2} \left(\frac12 + \alpha\right) + \frac rd \frac{N-k}N =: u_k
\end{align*}
where the left hand side is the probability of replacing a ball of type other than $i$ by one of type $i$ in state $p\in S_k$ (moving up one floor in the pyramid).
Similarly, we find that for $p \in S_k$, $k \in [N]$, we have
\begin{align*}
	\sum_{q\in S_{k-1}} P^{(N,r)}(p,q) &\le 2(1-r)\frac{k(N-k)}{N^2} \left(\frac12 - \alpha\right) + r\frac{d-1}d \frac kN =: d_k
\end{align*}
for the probability of replacing a ball of type $i$ by one of type other than $i$ in state $p\in S_k$ (moving down one floor in the pyramid).
Plugging this into \eqref{eq:updown}, we get
\begin{align}
	\sigma_{k-1} u_{k-1} \le \sigma_k d_k\text.
	\label{eq:increase}
\end{align}
All terms in \eqref{eq:increase} are strictly positive if $r > 0$.

Let $N$ be so that $\frac r{Nd} \ge 2\frac{1-r}{N^2}$ (we choose $r > 0$ later).
Then,
\begin{align*}
	u_k &\ge 2(1-r)\frac{k(N-k)}{N^2}\left(\frac12 + \alpha\right) + 2(1-r)\frac{N-k}{N^2}\\
	&\ge 2(1-r)\frac{(k+1)(N-k-1)}{N^2} \left(\frac12 + \alpha\right)
\end{align*}
where the last inequality uses $1 \ge \frac12 + \alpha$.
Similarly, we find that for $r\le \frac1d$ and $k \le N\left(1 - \frac{r}{\alpha}\right)$,
\begin{align*}
	d_k \le 2(1-r)\frac{k(N-k)}{N^2} \frac{1 - \alpha}{2}\text.
\end{align*}
Hence, with this bound on $k$, we have
\begin{align*}
	\frac{d_{k}}{u_{k-1}} \le \frac{1 - \alpha}{2\left(\frac12 + \alpha\right)} = \frac{1 - \alpha}{1 + 2\alpha} =: \beta\text.
\end{align*}
Thus, by \eqref{eq:increase}, $\frac{\sigma_{k-1}}{\sigma_k} \le \beta < 1$.
We have shown that the cumulative probability $\sigma_k$ of the states $S_k$ decreases at least as fast as the terms of the geometric series with parameter $\beta$ from some $k$ (close to $N$) downwards.

The maximal lottery for $R$ is the degenerate lottery with probability 1 on $i$.
For given $\delta,\tau > 0$, we are aiming for a lower bound on $N$ so that the probability on states with at least a $1-\delta$ fraction of balls of type $i$ in the stationary distribution $\pi$ is at least $1-\tau$.
That is,
\begin{align*}
 	\sum_{k=\lceil N(1-\delta)\rceil}^N \sigma_k \ge 1 - \tau\text.
\end{align*}
First observe that
\begin{align}
	\sum_{k\ge k_0} \beta^k = \beta^{k_0} \frac1{1 - \beta} \le \tau
		\label{eq:geometric}
\end{align}
for $k_0 \ge \frac{\log\left(\tau(1-\beta)\right)}{\log\beta}$.
For our bound, $N$ needs to be large enough so that there are at least $k_0$ integers in the interval $\{\left\lceil\left(1-\delta\right) N\right\rceil,\dots,\left\lfloor\left(1-\frac{r}\alpha\right)N\right\rfloor\}$.
The probability on states in $S_k$ with $k <\left(1-\delta\right)N$ will then be below $\tau$ by \eqref{eq:geometric} and the choice of $k_0$ (since the bound on $d_k$ assumes that $k \le N(1-\frac r\alpha)$).
Choosing $r \le \frac{\alpha\delta}2$ and
\begin{align*}
	N \ge \frac{k_0}{\delta - \frac r\alpha} \ge \frac{1}{\delta} \left\lceil\frac{\log\left(\tau(1-\beta)\right)}{\log\beta}\right\rceil
\end{align*}
achieves this.

In \Cref{ex:condorcetwinner}, there are three alternatives and 900 voters. 
Alternative 1 is a Condorcet winner as it is preferred to every other alternative by 600 of the voters ($\alpha = \frac23 - \frac12 = \frac16$, $\beta = \frac58$).
Suppose we want that at least $90\%$ of the balls in the urn are of type $1$ in at least $90\%$ of rounds ($\delta = 0.2$, $\tau = 0.1$).
Choosing $r = \frac{\alpha\delta}2 = \frac1{60}$, we need $N \ge 70$ balls in the urn. 
	These calculations suggest that, when a Condorcet winner exists, a reasonable choice of the parameters is $N \ge -\frac1\delta\log(\tau)$ and $\frac1N \le r \le \delta$.

\section{Discussion}
\label{sec:relatedwork}

Since the urn process described in this paper only depends on the comparison matrix $M_{\p}$ and the mutation rate $r$, it is connected to various problems unrelated to collective decision-making. In particular, the literature on equilibrium learning and evolutionary game theory has extensively studied dynamics based on payoff matrices and their convergence behavior.

\subsection{Equilibrium learning}
\label{sec:learning}

When interpreting $\tilde{M}_{\p}$ as a symmetric two-player zero-sum game and maximal lotteries as equilibrium strategies, our result can be phrased as a result about a learning procedure for equilibrium play. Such procedures have been extensively studied in game theory and, in particular for zero-sum games, a number of simple and attractive procedures have been proposed. The earliest of these is \emph{fictitious play} \citep{Brow51a,Robi51a} and its variant \emph{stochastic fictitious play} \citep{FuKr93a}.\footnote{\citet{HoSa02a} show that under stochastic fictitious play, players' strategies and beliefs converge to a Nash equilibrium in several classes of games, including two-player zero-sum games. While best-response dynamics are conceptually different from our urn process, their technical approach bears similarities to ours in that they use a deterministic process obtained as a solution to a differential equation to approximate a stochastic process.} 
More recently, the \emph{multiplicative weights update algorithm} \citep[e.g.,][]{FrSc99a,AHK12a} and \emph{regret matching} \citep[][]{HaMa00a,HaMa13a} have been celebrated in game theory, optimization, and machine learning. When translating the multiplicative weights update algorithm to our setting, one obtains a dynamic urn process, in which voters need to compare a drawn ball to all possible alternatives and adjust the distribution in the urn accordingly. It does not suffice to replace a single ball and the total number of balls does not remain constant. Also, the multiplicative weights update algorithm only guarantees convergence of the temporal average. The actual distribution does not converge, even for self-play in symmetric zero-sum games \citep{BaPi18a}.

A notable subarea of machine learning is concerned with \emph{multi-armed bandits}, a simple model of learning optimal sequential decisions when only very limited information is available \citep[see, e.g.,][]{BuCe12a,Sliv19a}. The theory of adversarial bandits is closely connected to learning in repeated multi-player games and it turns out that the prototypical algorithm for adversarial bandits, Exp3 (which stands for ``exponential-weight algorithm for exploration and exploitation''), bears some similarities to the urn process we describe in this paper. Exp3 can be formulated as an algorithm that learns an equilibrium strategy of a symmetric zero-sum game in self-play by iteratively updating a probability distribution merely based on the payoff associated with two actions randomly drawn from the current distribution. 
\citet{ACFS02a} prove strong bounds on the expected average regret and average regret achieved by Exp3 after a finite number of rounds, which imply that the temporal average of the distributions converges to a strategy close to an equilibrium. 
How close it gets to an equilibrium depends on a parameter that is roughly related to our mutation rate. Exp3 updates a probability distribution rather than the contents of a discrete urn and we are not aware of convergence results beyond the temporal average.

The literature on equilibrium learning often focusses on minimizing \emph{regret} rather than relative entropy with respect to an equilibrium distribution \citep[see, e.g.,][]{FoVo99a,ACFS02a}. In our context, the regret of the urn distribution at round $n$ is $\max_{i\in [d]} (\tilde{M}_R X^{(N,r)}(n,p_0))_i$. It follows from \Cref{thm:main} that for sufficiently large $n$, the regret is close to zero with high probability.
Our simulations show that the regret of the urn distribution converges faster than its relative entropy. This is interesting insofar as in order to approximately satisfy the desirable axiomatic properties of maximal lotteries discussed in \Cref{sec:approximateaxiomatics}, low regret is sufficient. It can be shown that a lottery has small regret if and only if it is a maximal lottery of a nearby preference profile. In other words, even if the urn distribution is still far from a maximal lottery, the distribution can perform almost as well as a maximal lottery. We have identified preference profiles where this effect is quite noticeable.

\subsection{Evolutionary Game Theory}

The \emph{replicator equation} in evolutionary game theory \citep[see, e.g.,][]{TaJo78a,ScSi83a,HoSi98a} describes how the distribution of different species changes continuously over time based on the individuals' fitnesses.
In its basic form, it states that the change in the relative frequency of a species equals the relative fitness of the species (that is, its fitness relative to the entire population) minus the change in the size of the entire population.
When the fitness depends linearly on the relative frequencies of the species and the population size is constant, the replicator equation defines the continuous deterministic process $y\colon \mathbb R_{\ge 0} \rightarrow \Delta$ with fitness function $f^{(r)}\colon \Delta\rightarrow\mathbb R^d$ below when setting $r$ to $0$.
When $r>0$, this process corresponds to a continuous and deterministic version of the urn process described in this paper (see \Cref{thm:continuous}).
\begin{equation}
	\begin{aligned}
			\frac d{dt} y(t) &= f^{(r)}(y(t))\text{ \quad and \quad} 
			y(0) = p_0 \text{,\quad where}\\
			f^{(r)}_i(p) &= 2(1 - r)p_i(\tilde Mp)_i + r\left(\frac1d - p_i\right)\text.
	\end{aligned}
	\label{eq:diffeq1}
\end{equation}

Solutions of this equation for $r=0$ are connected to \emph{evolutionary stable} distributions as introduced by \citet{MaPr73a}.
A distribution of species is evolutionary stable if its relative fitness exceeds that of every other distribution in a fixed neighborhood of it. 
Hence, evolutionary stable distributions are attractors of the dynamics defined by Equation~\eqref{eq:diffeq1} (with $r = 0$) in the sense that they are limit points of solutions when the initial distribution $p_0$ is in the respective neighborhood.
Mixed equilibria of zero-sum games such as Rock-Paper-Scissors usually fail to be evolutionary stable. As a consequence, results that prove convergence of dynamics to equilibrium strategies, either modify the underlying process or settle for weaker notions of convergence such as convergence of the temporal average.\footnote{
\citet{FoYo90a} argue that evolutionary stability is not an appropriate solution concept when stochastic events (such as random mutations or chance events in nature) affect the population.
They propose the \emph{stochastically stable set}, which is the smallest set of states such that for every neighborhood of it, with probability 1 the state is in that neighborhood all but a small fraction of the time.
\Cref{thm:main} shows that the stochastically stable set is contained in a small neighborhood of the set of maximal lotteries.
As the mutation rate $r$ goes to $0$, it converges to the set of maximal lotteries.
\citet{FoYo90a} show that the stochastically stable set is always non-empty and consists of those states that minimize a potential function.
}

In the following, we discuss five results that are closest to ours.

\begin{table}[tb]\footnotesize\centering
\makebox[\textwidth][c]{ 
\begin{tabular}{llllll}
\toprule
 & Model & Interaction & Mutations & Pop.~Size & Convergence\\
\midrule
Allesina et al.~(2011) & discrete & pairs, det. & no & fixed & ---$^a$\\ 
Knebel et al.~(2015) & continuous & pairs, det. & no$^b$ & fixed & temporal average\\ 
Laslier et al.~(2017) & discrete & pairs, det. & no & increasing & support of distribution\\ 
Laslier et al.~(2017) & discrete & triples, det. & no & increasing & distribution\\ 
Grilli et al.~(2017) & continuous & triples, stoch. & no & fixed & distribution\\ 
\midrule
\Cref{thm:main} & discrete & pairs, stoch. & yes & fixed & fraction of rounds\\ 
\Cref{cor:main} & discrete & pairs, stoch. & yes & fixed & temporal average\\
\Cref{thm:continuous} & continuous & pairs, det. & yes & fixed & distribution\\
\bottomrule
\end{tabular}
} 
\caption{Comparison of related models and results.\\
$a$: In simulations, \citet{AlLe11a} observe that the temporal average of their process comes close to a maximal lottery after a finite number of rounds. However, when the process is run long enough, the distribution will almost surely degenerate since there are no mutations.\\
$b$: While \citet{KWKF15a} consider a discrete process with mutations, the continuous process they study has no mutations.}% 
\label{tab:comparison}
\end{table}

\textbf{\citet{AlLe11a}} study the competition and coexistence of species in nature via a mathematical framework that is similar to our urn process. There is a fixed finite number of individuals, each of whom is assigned to some species at random. In each round, two randomly selected individuals interact. The superior species will replace the individual of the inferior. Which species is superior to which species is given in the form of a tournament graph, which can be represented by a binary comparison matrix. Interestingly, these tournaments are sampled from distributions that are obtained via multiple rankings of the species called ``limiting factors'', similar to the preferences of voters. Simulations with large populations (e.g., 25,000 individuals) then show that the relative frequencies of the species oscillate around the equilibrium strategy of the skew-comparison matrix. However, this phenomenon is an artifact of the population size and the limited time horizon. In the long run, as mentioned in \Cref{sec:intro}, all species but one will almost surely become extinct.\\
\textbf{\citet{KWKF15a}} study a dynamic process that involves quantum particles and is equivalent to a deterministic version of our urn process. 
Here, balls in the urn model correspond to bosons and alternatives to quantum states. 
The distribution of quantum states determines which states are condensates and are thus observed macroscopically.
Since the number of particles in such systems is typically large, they focus on a deterministic process with a continuum of particles as described in \Cref{sec:results}.
Leveraging a classic result from evolutionary game theory \citep[][Theorem 5.2.3]{HoSi98a}, they show that the \emph{temporal average} of this process converges to an equilibrium strategy (i.e., a maximal lottery) of the zero-sum game induced by the transition probabilities between quantum states. 
All states with probability zero in the equilibrium strategy are depleted; the fractions of the remaining states are bounded away from 0 for all times. 
Even though their model allows for mutations, \citeauthor{KWKF15a} neglect mutations when analyzing the continuous process, which may cause the process to cycle around the equilibrium strategy without converging to it.
Within our proof of \Cref{thm:main} (see \Cref{thm:continuous}), we show that the continuous process with mutations does converge (and not only its temporal average).\footnote{
\citet[][Supplementary Note 1]{KWKF15a} argue that the discrete process with mutations is well-approximated by the continuous process if the number of particles is large and mutations become vanishingly unlikely.
Hence, they conclude that the temporal average of the discrete process converges to an equilibrium strategy, which is in the spirit of \Cref{cor:main}. 
Our understanding is that their arguments are heuristic and not intended to provide a rigorous derivation of this result.
In particular, the arguments do not seem to use that mutations happen with non-zero probability.
Without mutations, however, the discrete process almost surely enters a state with a degenerate distribution.
}
In earlier work, \citet{KKWF13a} have connected the survival and extinction of states to the Pfaffian of the transition matrix. This is reminiscent of a statement by \citet{Kapl95a} about the support of equilibrium strategies in symmetric zero-sum games. 
\citet{RMF06a} study the extinction probabilities for three states with cyclical dominance (``rock-paper-scissors'') for finite populations.\\
\textbf{\citet{LaLa13a}} consider a discrete urn process that is similar to ours, but in which the number of balls in the urn increases over time. 
Two balls are drawn at random and a binary comparison matrix specifies which alternative wins against which alternative (this could be seen as a single voter with possibly intransitive preferences in our model). 
Rather than replacing the losing ball, a new ball of the same type as the winning ball is added to the urn. They show that the distribution in the urn does not converge unless one alternative beats all alternatives (which corresponds to the Condorcet winner case). 
However, the fraction of alternatives not contained in the support of the maximal lottery of the skew-comparison matrix goes to zero. 
Their main results concerns a process in which \emph{three} balls are drawn from the urn. Whenever one of three balls beats both other balls, a new ball of the same type is added to the urn. Otherwise, one of the three types is chosen at random and a ball of that type is added. They prove that the distribution in the urn converges towards the (unique) maximal lottery of the skew-comparison matrix. 
Since the number of balls in the urn increases, convergence is generally very slow.\\
\textbf{\citet{GBMA17a}} consider a dynamic process in population biology to explain the stable coexistence of multiple species. Based on \citeauthor{LaLa13a}'s findings, \citeauthor{GBMA17a} adapt the replicator equation to interactions of triples of individuals. In contrast to \citeauthor{LaLa13a}, they keep the number of individuals constant and do not require the comparison matrix to be binary.
They show that with a continuum of individuals, this process converges to an equilibrium strategy of the skew-comparison matrix.
For a finite number of individuals, permanent coexistence of multiple species is a probability zero event.
However, they argue that interactions of three or more individuals can prolong coexistence compared to pairwise interactions.

\medskip

Without mutations (i.e., $r = 0$), the deterministic process described by the differential equation~\eqref{eq:diffeq1} does not, in general, converge, but only approaches an orbit of constant entropy relative to a zero of the fitness function $f^{(0)}$.
When introducing mutations, the limiting behavior of the process changes qualitatively (see \Cref{fig:3cycle+condorcetloser-cont}). 
As \Cref{thm:continuous} 
shows, it then converges to a zero of the fitness function.
A similar observation has already been made by \citet[][Theorem 2.8]{Hofb11a}.

\begin{figure}[tb]

	\centering
	\includegraphics[width=0.49\textwidth]{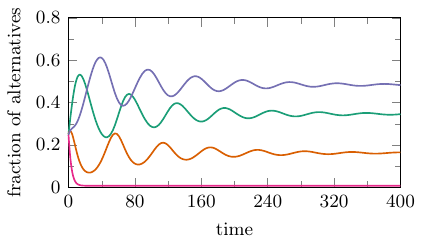}
	\includegraphics[width=0.49\textwidth]{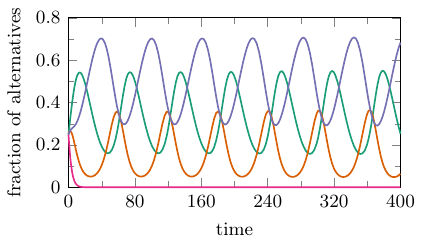}
	\caption{The continuous deterministic process $y(t)$ solving \Cref{eq:diffeq1} for the profile in \Cref{ex:3cycle+condorcetloser}
	with $r = 0.01$ on the left and $r = 0$ on the right. 
	For strictly positive $r$, $y(t)$ converges to a zero of $f^{(r)}$ (see \Cref{thm:continuous}). For $r = 0$, it approaches an orbit of constant entropy relative to a zero of $f^{(r)}$. 
	}
	\label{fig:3cycle+condorcetloser-cont}
\end{figure} 

\begin{restatable}{theorem}{continuous}\label{thm:continuous}
	Let $f^{(r)}$ and $y$ be defined as in \Cref{eq:diffeq1}. 
	If $r > 0$, $f^{(r)}$ has a unique zero $p^{(r)}$ and $y(t)$ converges to $p^{(r)}$ as $t\rightarrow\infty$.
	Moreover, if $r$ goes to 0, then $p^{(r)}$ converges to $\ml(\p)$ in Hausdorff distance.
\end{restatable}

\Cref{tab:comparison} summarizes the key differences between the above mentioned results and ours. In comparison, the main contribution of our work is that we are able to show for a discrete (rather than continuous) process based on stochastic (rather than deterministic) interactions between pairs (rather than triples) that the distribution in the urn is close to a maximal lottery most of the time (rather than convergence of the temporal average).
Methodologically, the approach we take to cope with the discrete process is related to that of \citet{BeWe03a}, who study more general population processes in $n$-player games.\footnote{In \citeauthor{BeWe03a}'s model, there is a population of $N$ individuals for each player, and each individual plays a pure strategy.
In each round, one individual can update their strategy based on the distributions of pure strategies of all other individuals.
An update rule induces a deterministic process described by a differential equation similar to~\eqref{eq:diffeq1} below.
They show that if $N$ is large, the distributions of strategies among the individuals of each role in this stochastic process approximate the deterministic process described by the differential equation.
Our setting corresponds to a symmetric two-player zero-sum game and an update rule based on the comparison matrix $\tilde M$.
The special properties of this instance allow us to make more precise statements about the behavior of the deterministic process, and, thus, of the stochastic process for large $N$. In particular, we show that the deterministic process converges and that its limit approximates a maximal lottery.
}

We believe that \Cref{thm:continuous} as well as \Cref{thm:main} and \Cref{cor:main} are of relevance to the natural sciences. In particular, a discrete model may describe the aforementioned natural phenomena more accurately than continuous ones. 
As \Cref{cor:main} shows, the expectation of the discrete process with a large number of individuals is a good approximation of the continuous process.
Furthermore, the observation that convergence is only guaranteed if mutations occur with small probability and the number of individuals is large enough seems noteworthy.

\section*{Acknowledgments}{\footnotesize% 
	This material is based on work supported by the Deutsche Forschungsgemeinschaft under grants {BR~2312/11-1}, {BR~2312/11-2}, {BR~2312/12-1}, and {BR~5969/1-1}. The authors thank Stefano Allesina, Stergios Athanasoglou, Vincent Conitzer, Javier Esparza, Erwin Frey, Drew Fudenberg, Philipp Geiger, Umberto Grandi, Matthias Greger, Josef Hofbauer, Sean Horan, Johannes Knebel, Jean-François Laslier, Patrick Lederer, Hervé Moulin, Noam Nisan, Robert Schapire, Omer Tamuz, Nicolas Vieille, Jörgen Weibull, Peyton Young, and the participants of the COMSOC video seminar (March 2021), the International Conference on ``New Directions in Social Choice'' (St. Petersburg, July 2021), the Hausdorff Center for Mathematics Symposium (Bonn, August 2021), the seminar of the Center of Economic Research of ETH Z\"urich (Zurich, November 2021), the Nobel Symposium ``One Hundred Years of Game Theory'' (Stockholm, December 2021), the Microeconomics Research Seminar at the University of Hamburg (Hamburg, April 2022), the Economics Seminar at the University of Milano-Bicocca (Milan, May 2022), the Hi!Paris Symposium on Artificial Intelligence and the Social Sciences (Paris, June 2022), the 16th Meeting of the Society of Social Choice and Welfare (Mexico City, June 2022), and the Economics Seminar at Bielefeld University (Bielefeld, June 2022) for stimulating discussions and encouraging feedback.\par}

\newpage
\appendix
\section*{APPENDIX: Proofs}

As guidance for the reader, we outline the main steps in the proof of \Cref{thm:main}.
Fix any $\delta,\tau > 0$.

In \Cref{sec:vectorfield}, we consider, for any $p\in \Delta^{(N)}$, the expected value of $N\left(X^{(N,r)}(k+1,p) - X^{(N,r)}(k,p)\right)$, which, conditional on $X^{(N,r)}(k,p)$, is independent of $k$ since $X^{(N,r)}$ is a time-homogeneous Markov process.
Moreover, it is independent of $N$ since the probability of replacing a ball of type $j$ by one of type $i$ is independent of $N$.
Hence, these expected values induce a continuous function $f^{(r)}\colon\Delta\rightarrow\mathbb R^d$.
We consider $g^{(r)}\colon\Delta\rightarrow\mathbb R^d$ with $g^{(r)}(p) = p + \frac12f^{(r)}(p)$ and show that it maps to $\Delta$.
If $r > 0$, $g^{(r)}$ has a unique fixed-point $p^{(r)}$ (a zero of $f^{(r)}$), which is close to some lottery in $\ml(\p)$ for any small enough $r$.
We choose $r_0$ so that $p^{(r)}$ is no more than $\frac\delta2$ away from $\ml(\p)$ for all $0 < r \le r_0$.
Fixing such an $r$, let $p^*\in\ml(\p)$ be a maximal lottery which is within $\delta$ of $p^{(r)}$.

\Cref{sec:deterministicprocess} studies the following differential equation with $p\in\Delta$, $t\in \mathbb R_{\ge 0}$, and $y(\cdot,p)\colon\mathbb R_{\ge 0}\rightarrow\Delta$.
\begin{align}
	\begin{aligned}
		\frac d{dt}y(t,p) &= f^{(r)}(y(t,p))\\
		y(0,p) &= p
	\end{aligned}
	\label{eq:diffeq0}
\end{align}
A solution to \eqref{eq:diffeq0} is a \emph{deterministic} process that can be interpreted as the \emph{stochastic} process we consider with a continuum of balls.
We show that the unique solution $y^{(r)}(\cdot,p)$ of \eqref{eq:diffeq0} converges to $p^{(r)}$ for any initial state $p\in\Delta$ as $t$ goes to infinity and the convergence is uniform in $p$.
This is done by showing that the entropy of $p^{(r)}$ relative to $y^{(r)}(t,p)$ decreases monotonically at a rate proportional to the square of the distance between $p^{(r)}$ and $y^{(r)}(t,p)$.

\Cref{sec:propertiesofthediscreteprocess} relates the discrete-time stochastic process $X^{(N,r)}$ to the continuous-time deterministic process $y^{(r)}$.
To this end, we extend the former to the real time axis by letting $\bar X^{(N,r)}(t,p) = X^{(N,r)}(k,p)$ for $t\in[\frac{k-1}N,\frac kN)$.
Given any $T > 0$, one can show that with probability close to 1, $\bar X^{(N,r)}$ \emph{approximately} satisfies the integral equation corresponding to \eqref{eq:diffeq0} for $t$ between $0$ and $T$ and uniformly in $p\in\Delta^{(N)}$ if $N$ is large.
Using Gr\"onwall's inequality, we show that with probability close to 1, $\bar X^{(N,r)}(t,p)$ and $y^{(r)}(t,p)$ are close to each other for all $t$ from $0$ to $T$.\footnote{In the language of functional analysis, this step corresponds to an approximation of an operator semi-group. Consider the operators $\Gamma(t)$ on probability measures on $\Delta$ induced by mapping $p\in\Delta$ to $y^{(r)}(t,p)$. Then $\{\Gamma(t)\colon t\ge 0\}$ is an operator semi-group (that is, $\Gamma(s+t) = \Gamma(s)\Gamma(t)$). On $\Delta(\Delta^{(N)})$, we approximate $\Gamma(t)$ by $(P^{(N,r)})^{Nt}$.}
However, for $t$ larger than $T$, they may (and almost surely will) be arbitrarily far apart.

To deal with this, we partition the time axis into consecutive intervals of length $T$ and synchronize the deterministic process with the stochastic process at the beginning of each interval.
More precisely, since $y^{(r)}(t,p)$ converges to $p^{(r)}$ as $t$ goes to infinity uniformly in $p$, we can find $T > 0$ such that $y^{(r)}(t,p)$ is no more than $\frac\delta 4$ away from $p^{(r)}$ for all but possibly a $1-\frac\tau2$ fraction of the interval $[0,T]$ for all $p$.
Moreover, we can choose $N$ large enough so that with probability at least $1-\frac\tau2$, the distance between $\bar X^{(N,r)}$ and $y^{(r)}$ is less than $\frac\delta 4$ for all $t$ in an interval of length $T$ provided both processes start at the same point at the beginning of the interval.
We chop up the time axis into intervals $[0,T]$, $[T,2T]$, $\dots$.
On the interval $[(k-1)T,kT]$, we compare $\bar X^{(N,r)}(t,p)$ to $y^{(r)}(t - (k-1)T, \bar x_{k-1})$, where $\bar x_{k-1} = X^{(N,r)}((k-1)T,p)$.
That is, we reset $y^{(r)}$ to the position of $\bar X^{(N,r)}$ at the beginning of the interval.
In those intervals where the distance between both processes is never more than $\frac\delta 4$, $\bar X^{(N,r)}$ is no more than $\frac\delta 4 + \frac\delta 4 = \frac\delta2$ away from $p^{(r)}$ for all but a $\frac\tau2$ fraction of the interval.
By the choice of $N$, the union of those intervals is almost surely at least a $1-\frac\tau2$ fraction of the time axis.
Summing over all intervals, this is enough to conclude that $\bar X^{(N,r)}$ is no more than $\frac\delta2$ away from $p^{(r)}$ at least a $1-\tau$ fraction of the time.
Since $p^{(r)}$ is no more than $\frac\delta2$ away from $p^*$, we can get the same conclusion with $\delta$ in place of $\frac\delta2$ and $p^*$ in place of $p^{(r)}$.
Translating this statement back to $X^{(N,r)}$ gives the first part of \Cref{thm:main}.

\section{A Continuous Vector Field Induced by the Markov Chain}\label{sec:vectorfield}

In this section, we define a continuous mapping from $\Delta$ to $\Delta$ based on the \emph{expected} urn distribution in the subsequent round for each state of the Markov chain. 
We then show that this mapping admits a unique fixed-point corresponding to an approximate maximal lottery.

Recall that $\{X^{(N,r)}(n,p_0)\colon n\in\mathbb N_0\}$ is a discrete-time, time-homogeneous Markov chain with state space $\Delta^{(N)}$ and transition probability matrix 
\begin{align*}
	P^{(N,r)}(p,p') =
	\begin{cases}
		2(1-r) p_ip_jM(i,j) + \frac rd p_j\qquad&\text{if }i \neq j\\
		2(1-r) \sum_{k = 1}^d p_k^2 + \frac rd &\text{if } i = j\\
	\end{cases}
\end{align*}
for $p\in \Delta^{(N)}$ and $p' = p +\frac{e_i}N-\frac{e_j}N$ for $i,j\in[d] = \{1,\dots,d\}$ with $p'\in \Delta^{(N)}$.
All other transition probabilities are 0.
If $r > 0$, it is irreducible and aperiodic and, thus, admits a unique stationary distribution in $\Delta(\Delta^{(N)})$, a probability distribution over urn distributions. 
We omit writing the initial state $p_0$ whenever it is convenient. 

For $i \in [d]$, we calculate the expected change in the $i$th component of $X^{(N,r)}$ times $N$ given that $X^{(N,r)}$ is in state $p\in \Delta^{(N)}$.
\begin{align*}
	&\phantom{=\text{ }} N\ev{X^{(N,r)}_i(n+1) - X^{(N,r)}_i(n)\mid X^{(N,r)}(n)= p}\\
	&= N\sum_{p'\in \Delta^{(N)}} (p'_i - p_i) P^{(N,r)}(p,p') \\
	&= 2(1-r) \sum_{j \neq i} p_ip_j \left(M(i,j) - M(j,i)\right) + \frac rd \sum_{j\neq i} \left(p_j - p_i\right)\\
	&= 2(1-r) p_i\sum_{j \neq i} \tilde M(i,j)p_j + \frac rd\left(1-p_i - (d-1)p_i\right)\\
	&= 2(1-r) p_i (\tilde Mp)_i + r\left(\frac1d - p_i\right)
\end{align*}
For the last equality, recall that $\tilde M(i,i) = 0$ since $\tilde M$ is skew-symmetric.

Based on this, we define the continuous function $f^{(r)}\colon \Delta \rightarrow \mathbb R^d$ with
\begin{align}
	f^{(r)}_i(p) = 2(1-r) p_i (\tilde Mp)_i + r\left(\frac1d - p_i\right)\text.
	\label{eq:f}
\end{align}

Let $g^{(r)}\colon\Delta\rightarrow\Delta$ with $g^{(r)}(p) = p + \frac12f(p)$ for $p\in\Delta$.
We show that $g^{(r)}$ is well-defined (that is, indeed maps to $\Delta$) and has a fixed-point.
If $r > 0$, this fixed-point is unique and we denote it by $p^{(r)}$.
As $r$ goes to 0, $p^{(r)}$ converges to the set of maximal lotteries for the profile $\p$ that induces $\tilde M$.
We note that if $r = 0$, $g^{(r)}$ has a unique fixed-point if and only if there is a unique maximal lottery.

\begin{lemma}\label{lem:stationary}
	For $r > 0$, $g^{(r)}$ has a unique fixed-point $p^{(r)}$.
	Moreover, for every $\delta > 0$, there is $r_0$ so that $p^{(r)}\in B_\delta(\ml(\p))$ for all $r \le r_0$.
\end{lemma}

\begin{proof}
	We verify that $g^{(r)}$ maps to $\Delta$.
	For all $p\in\Delta$, $\sum_{i \in[d]} f^{(r)}_i(p) = 2(1-r) p^\intercal\tilde M p + r\left(1 - \sum_{i\in[d]} p_i\right) = 0$ since $\tilde M$ is skew-symmetric and $p\in\Delta$.
	Moreover,
	\begin{align*}
		f^{(r)}_i(p) = 2(1-r)p_i \underbrace{(\tilde Mp)_i}_{\ge -1} + r\left(\frac 1d - p_i\right) \ge -2p_i\text.
	\end{align*}
	Thus,
	\begin{align*}
		g^{(r)}_i(p) \ge p_i + \frac12 (-2p_i) \ge 0\text.
	\end{align*}
	It follows that $g^{(r)}$ maps to $\Delta$.
	Moreover, $g^{(r)}$ is continuous since $f^{(r)}$ is continuous.
	Hence, by Brouwer's Theorem, $g^{(r)}$ has a fixed point $p^{(r)}$.

	Now let $r > 0$.
	Then, for all $p\in\Delta$ with $f^{(r)}(p) = 0$, we have for all $i\in[d]$, $p_i > 0$ since $p_i = 0$ implies $f^{(r)}_i(p) = r\frac1d > 0$.
	Hence, we can rewrite $f^{(r)}(p) = 0$ as follows: for all $i\in[d]$,
	\begin{align}
		2(1-r)(\tilde M p)_i = r\left(1 - \frac1{p_i d}\right)
		\label{eq:stationary2}
	\end{align}
	
	To show that $f^{(r)}$ has a unique zero, assume that $f^{(r)}(p) = f^{(r)}(q) = 0$ for $p,q \in\Delta$.
	We have 
	\begin{align*}
		0 &= 2(1-r)\left(p^\intercal\tilde Mq + q^\intercal\tilde M p\right)\\
		&= 2(1-r)\sum_{i\in[d]} p_i \left(\tilde Mq\right)_i + q_i\left(\tilde Mp\right)_i\\
		&\overset{\eqref{eq:stationary2}}{=} r\sum_{i \in[d]} p_i\left(1-\frac1{q_id}\right) + q_i\left(1 - \frac1{p_id}\right)\\
		&= \frac rd \sum_{i\in[d]} \frac{p_iq_i - p_i}{q_i} + \frac{p_iq_i - q_i}{p_i}\\
		&= -\frac rd \sum_{i\in[d]} \frac{\left(p_i - q_i\right)^2}{p_i q_i} \le - \frac rd\lvert p - q\rvert_2^2
	\end{align*}
	where the first equality uses the skew-symmetry of $\tilde M$ (hence, $p^\intercal\tilde M q = - q^\intercal\tilde Mp$), the third equality follows from \eqref{eq:stationary2} and the fact that $p$ and $q$ are zeros of $f^{(r)}$, and the last two are algebra.
	($\lvert\cdot\rvert_2$ denotes the $L^2$-norm.)
	This sequence of equalities implies that $p = q$.
	Hence, $p^{(r)}$ is the unique zero of $f^{(r)}$ for $r > 0$.
	Since every fixed-point of $g^{(r)}$ is a zero of $f^{(r)}$, $g^{(r)}$ has a unique fixed-point.
	
	For the last statement, let $\delta > 0$.
	By \eqref{eq:stationary2}, for all $r > 0$ and $i\in[d]$,
	\begin{align}
		\left(\tilde M p^{(r)}\right)_i = \frac r{2(1-r)}\left(1 - \frac1{p_i^{(r)}d}\right) \le \frac r{2(1-r)}\text.
		\label{eq:mllimit}
	\end{align}
	Suppose for every $r_0 > 0$, there is $r < r_0$ so that $p^{(r)} \not\in B_\delta(\ml(\p))$.
	Then we can find a sequence $(r_n)$ going to 0 so that $p^{(r_n)}\not\in B_\delta(\ml(\p))$ for all $n$.
	By passing to a subsequence, we may assume that $p^{(r_n)} \rightarrow p\not\in B_\delta(\ml(\p))$.
	But from \eqref{eq:mllimit} it follows that $\tilde M p \le 0$ so that $p \in \ml(\p)$, which is a contradiction.
\end{proof}

\section{Properties of the Deterministic Process}\label{sec:deterministicprocess}

In this section, we study a deterministic version of the stochastic process described by the Markov chain. We thus have a continuum of balls and continuous time, and show that this process converges to the unique fixed-point identified in the previous section.

Function $f^{(r)}$ defined in \Cref{eq:f} gives rise to a (first-order ordinary) differential equation for continuously differentiable functions from $[0,\infty)$ to $\Delta$, that is, functions in $\mathcal C^1([0,\infty),\Delta)$.
For $y\in \mathcal C^1([0,\infty),\Delta)$ and $p_0\in\Delta$, consider
\begin{align}
	\begin{aligned}
		\frac d{dt} y(t) &= f^{(r)}(y(t))\\
		y(0) &= p_0
	\end{aligned}
	\label{eq:diffeq}
\end{align}
We show that \eqref{eq:diffeq} has a unique global solution $y^{(r)}$ for all $r > 0$ and $p_0\in\Delta$.
Moreover, this solution converges to the zero $p^{(r)}$ of $f^{(r)}$ as $t$ goes to infinity.
Since $r$ remains fixed throughout this section, we frequently omit the superscript $(r)$.

The proof that \eqref{eq:diffeq} has a unique local solution with values in $\mathbb R^d$ is standard.
Only the fact that the solution does not leave the domain $\Delta$ of $f$ and can, thus, be extended to a global solution requires attention.

\begin{lemma}\label{lem:diffeqexistence}
	For every $p_0 \in\Delta$, \eqref{eq:diffeq} has a unique solution $y\in \mathcal C^1([0,\infty),\Delta)$ with $y(0) = p_0$.
\end{lemma}

\begin{proof}
	Note that $f$ is Lipschitz-continuous in a neighborhood of $\Delta$. 
	It follows from the Picard-Lindel\"of Theorem that for any $t_0\in[0,\infty)$ and $p\in\Delta$, the system
	\begin{align}
		\begin{aligned}
			\frac d{dt} y(t) &= f(y(t))\\
			y(t_0) &= p
		\end{aligned}
	\end{align}
 	has a unique local solution, that is, a solution $y\in\mathcal C^1((t_0-\epsilon,t_0+\epsilon),\mathbb R^d)$.
	
	We observe that $y$ maps to $\Delta$.
	First, by the same arguments as in the proof of \Cref{lem:stationary}, we have 
	\begin{align*}
		\frac d{dt}\sum_{i\in[d]} y_i(t) = \sum_{i\in[d]} f_i(y(t)) = 0
	\end{align*}
	whenever $y(t)\in\Delta$.
	Second, if $y_i(t) = 0$, then $\frac d{dt} y_i(t) = f_i(y(t)) > 0$.
	Hence, $y(t) \in\Delta$ for all $t\in(t_0 - \epsilon,t_0+\epsilon)$.
	Since $t_0\in[0,\infty)$ was arbitrary, it follows that $y$ can be uniquely extended to a global solution in $\mathcal C^1([0,\infty),\Delta)$.
\end{proof}

Denote by $y^{(r)}(t,p_0)\in\mathcal C^1([0,\infty),\Delta)$ the unique solution to \eqref{eq:diffeq} with $y^{(r)}(0,p_0) = p_0$.
We will sometimes suppress the argument $p_0$ when it is clear from the context.

We want to show that if $r > 0$, $y^{(r)}(t,p_0)$ converges to the zero $p^{(r)}$ of $f^{(r)}$ as $t$ goes to infinity.
Moreover, the convergence is uniform in $p_0$.
The proof of this fact in \Cref{lem:continuoustimeconvergence} uses the relative entropy (aka the Kullback–Leibler Divergence) of $p,q\in\Delta$, which is defined as
\begin{align*}
	D(p\mid q) = \sum_{i\in[d]} p_i \log\left(\frac{p_i}{q_i}\right)\text.
\end{align*}
Moreover, the following lower bound on the relative entropy will be helpful \citep[see, e.g.,][Lemma 11.6.1]{CoTh06a}.
\begin{lemma}\label{lem:relativeentropybound}
	For all $p,q\in\Delta$,
	\begin{align*}
		D(p\mid q) \ge \frac1{2\log 2} \left\lvert p - q\right\rvert^2\text.
	\end{align*}
\end{lemma}

To ease notation, we write $\chi_S$ for the indicator function of a set $S\subset\mathbb R^d$ and $\bar\chi_S = 1-\chi_S$ for the indicator function of the complement of $S$.

\begin{lemma}\label{lem:continuoustimeconvergence}
	Let $r > 0$.
	Then, 
	\begin{align*}
		\lim_{t\rightarrow\infty}\sup\left\{\left\lvert y^{(r)}(t,p_0) - p^{(r)}\right\rvert\colon p_0\in\Delta\right\} = 0\text.
	\end{align*}
\end{lemma}

\begin{proof}
	Fix $p_0$ in the interior of $\Delta$ and write $y = y(\cdot, p_0)$.
	We show that the entropy of $p^{(r)}$ relative to $y(t)$ decreases at a rate of at least $\frac r{d\sqrt{d}}\left|p^{(r)} - y(t)\right|_2^2$.
	\begin{align*}
		\frac d{dt} D(p^{(r)} \mid y(t)) &= \frac d{dt} \sum_{i\in [d]} p^{(r)}_i \log\left(\frac{p^{(r)}_i}{y_i(t)}\right) = -\sum_{i\in [d]} p^{(r)}_i \frac{\frac d{dt}y_i(t)}{y_i(t)}\\
		&\overset{\rom{1}}{=} -\sum_{i\in [d]} p^{(r)}_i \frac{f_i(y(t))}{y_i(t)} \\
		&= -\sum_{i \in[d]} p^{(r)}_i \frac{2(1-r)y_i(t) (\tilde M y(t))_i + r \left(\frac1d - y_i(t)\right)}{y_i(t)}\\
		&= -2(1-r)\sum_{i\in [d]} p^{(r)}_i (\tilde M y(t))_i - r\sum_{i\in [d]} p^{(r)}_i \left(\frac{1}{y_i(t) d} - 1\right)\\
		&\overset{\rom{2}}{=} 2(1-r)\sum_{i\in [d]} y_i(t) (\tilde M p^{(r)})_i - r\left(\sum_{i\in [d]} \frac{p^{(r)}_i}{y_i(t)d} - 1\right)\\
		&\overset{\rom{3}}{=} \sum_{i\in [d]} y_i(t) r\left( 1 -\frac1{p^{(r)}_i d} \right) - r\left(\sum_{i\in [d]} \frac{p^{(r)}_i}{y_i(t)d} - 1\right)\\
		&= r\left(2 - \frac1d\sum_{i\in [d]} \frac{y_i(t)}{p^{(r)}_i} + \frac{p^{(r)}_i}{y_i(t)}\right)\\
		&\overset{\rom{4}}{=} -\frac rd \sum_{i\in [d]} \frac{(p^{(r)}_i - y_i(t))^2}{p^{(r)}_iy_i(t)} \le -\frac r{d\sqrt{d}} \left\lvert p^{(r)} - y(t)\right\rvert^2 
	\end{align*}
	Here, $\rom{1}$ follows from the fact that $y$ satisfies \eqref{eq:diffeq}, $\rom{2}$ uses the skew-symmetry of $\tilde M$ and $\sum_{i\in[d]} p^{(r)}_i = 1$, $\rom{3}$ uses \eqref{eq:stationary2}, and $\rom{4}$ uses $a^2 + b^2 = (a + b)^2 -2ab$ for any $a,b\in\mathbb R$.

	For $t \ge t_0 \ge 0$, we have
	\begin{align*}
		0 \le D(p^{(r)} \mid y(t)) = D(p^{(r)}\mid y(t_0)) + \int_{t_0}^t \frac d{ds} D(p^{(r)} \mid y(s))ds \le D(p^{(r)} \mid y(t_0))\text.
	\end{align*}
	Combining this with the sequence of equalities above, we see that
	 
	\begin{align}
		0 \le \frac rd\sum_{i = 1}^d \int_{t_0}^t \frac{(p^{(r)}_i - y_i(s))^2}{p^{(r)}_iy_i(s)} ds = -\int_{t_0}^t \frac d{ds} D(p^{(r)}\mid y(s)) ds \le D(p^{(r)}\mid y(t_0))\text.
		\label{eq:relativeentropybound}
	\end{align}

	We want to prove that $y(t,p_0)$ converges to $p^{(r)}$ uniformly in $p_0$ as $t$ goes to $\infty$.
	That is, for all $\epsilon > 0$, there exists $T > 0$ such that for all $t \ge T$ and all $p_0 \in \Delta$, $\lvert y(t,p_0) - p^{(r)}\rvert < \epsilon$.
	To this end, first note that if $y_i(t,p_0) < \frac r{4d}$, then
	\begin{align*}
		\frac d{dt} y_i(t,p_0) = 2(1-r)y_i(t,p_0) \underbrace{\left(\tilde My(t,p_0)\right)_i}_{\ge -1} + r\left(\frac 1d - y_i(t,p_0)\right) \ge - \frac r{2d} + \frac rd  \ge \frac{r}{2d}\text.
	\end{align*}
	Hence, for all $p_0 \in\Delta$, $i \in[d]$, and $t \ge 1$, $y_i(t,p_0) \ge \frac r{4d}$.
	We can, thus, upper bound $D(p^{(r)} \mid y(t,p_0))$ for all $p_0 \in\Delta$ and $t \ge 1$ by $C = \max_{p \in \Delta^r} D(p^{(r)}\mid p) < \infty$, where $\Delta^r = \{p \in\Delta\colon p_i \ge \frac r{4d} \text{ for all } i\in[d]\}$.

	Now we prove uniform convergence in $p_0$. 
	Let $\epsilon > 0$.
	It follows from \eqref{eq:relativeentropybound} with $t_0 = 1$ that given $\delta > 0$, for all $p_0 \in\Delta^r$,
	\begin{align*}
		\int_{t \ge 1} \bar\chi_{B_\delta(p^{(r)})}(y(t,p_0))dt \le \frac {Cd\sqrt{d}}{r\delta^2}\text.
	\end{align*}
	 
	Hence, for every $p_0 \in\Delta^r$, we can find $t_0(p_0,\delta) \in [1,1 + \frac{Cd\sqrt{d}}{r\delta^2}]$ such that 
	\begin{align}
		\left\lvert y(t_0(p_0,\delta),p_0) - p^{(r)}\right\rvert < \delta\text.
		\label{eq:deltabound}
	\end{align}
	Using the estimate $\log(x - \delta) \ge \log(x) - \frac\delta{x-\delta}$ for the last inequality, we find that
	\begin{align*}
		D(p^{(r)}\mid y(t_0(p_0,\delta))) &\le \sum_{i\in[d]} \log\left(\frac{p^{(r)}_i}{p^{(r)}_i-\delta}\right) = \sum_{i\in[d]}\log\left(p_i^{(r)}\right) - \log\left(p_i^{(r)} - \delta\right)\\
		&\le \sum_{i\in[d]} \frac{\delta}{p_i^{(r)}-\delta} \le \delta C'
	\end{align*}
	where $C' = 2d\max\{\frac1{p_i^{(r)}}\colon i\in[d]\}$ if $\delta\in(0,\frac12\min\{p_i^{(r)}\colon i\in[d]\})$.
	
	We use this bound and the fact that the relative entropy is non-increasing in $t$ to show that $\lvert y(t,p_0) - p^{(r)}\rvert < \epsilon$ for $t \ge t_0(p_0,\delta)$ for sufficiently small $\delta$.
	By \Cref{lem:relativeentropybound}, we have for all $p \in\Delta$, $D(p^{(r)}\mid p) \ge \frac1{2\log 2} \lvert p^{(r)} - p\rvert^2$.
	Hence, $\lvert p^{(r)} - y(t,p_0)\rvert \le \sqrt{2\log(2) \delta C'}$ for $t \ge t_0(p_0,\delta)$.
	Recalling that $t_0(p_0,\delta) \le 1 + \frac{Cd\sqrt{d}}{r\delta^2} =: T$, we have for $\delta \in(0,\frac{\epsilon^2}{2\log(2)C'})$ that $\lvert p^{(r)} - y(t,p_0)\rvert < \epsilon$ for all $t \ge T$ and $p_0 \in \Delta$.
	Since $\epsilon$ was arbitrary, this proves uniform convergence.
\end{proof}

The next lemma states that for any $\delta > 0$, if the process $y^{(r)}$ starts sufficiently close to $p^{(r)}$, it will never get further than $\delta$ away from $p^{(r)}$.

\begin{lemma}\label{lem:continuoustimestaysclose}
	Let $r > 0$ and $\delta > 0$.
	Then, there is $\eta > 0$ such that
	\begin{align*}
		\sup\left\{\left\lvert y^{(r)}(t,p) - p^{(r)}\right\rvert\colon t\ge 0,\, p\in B_\eta(p^{(r)})\right\} < \delta
	\end{align*}
\end{lemma}
\begin{proof}
	Recall that $p^{(r)}_i > 0$ for all $i\in[d]$.
	By \Cref{lem:relativeentropybound}, if $p\not\in B_\delta(p^{(r)})$, then $D(p^{(r)}\mid p) \ge \frac1{2\sqrt{2}}\delta^2 =: C$.
	Since $D(p^{(r)}\mid\cdot)$ is continuous on the interior of $\Delta$ and $D(p^{(r)}\mid p^{(r)}) = 0$, there is $\eta > 0$ such that $D(p^{(r)}\mid p) < C$ for all $p\in B_\eta(p^{(r)})$.
	In the proof of \Cref{lem:continuoustimeconvergence}, we have seen that $D(p^{(r)}\mid y^{(r)}(t,p))$ is non-increasing in $t$.
	Hence, for $p\in B_\eta(p^{(r)})$, it follows that $\lvert y^{(r)}(t,p) - p^{(r)}\rvert < \delta$ for all $t \ge 0$.
\end{proof}

Summarizing \Cref{lem:stationary}, \Cref{lem:diffeqexistence}, and \Cref{lem:continuoustimeconvergence}, we get the following theorem.

\continuous*

\begin{remark}[Drawing without replacement]\label{rem:proofadaption1}
	For the urn process with drawing without replacement, $f^{(r)}$ as derived in \Cref{sec:vectorfield} would depend on $N$.
	The solution $y^{(r)}$ of the differential equation \eqref{eq:diffeq} and the unique zero $p^{(r)}$ of $f^{(r)}$ would, thus, also depend on $N$.
	The previous lemmas carry over to this case with the straightforward adaptations.
\end{remark}

\section{Properties of the Stochastic Process}\label{sec:propertiesofthediscreteprocess}

In this section, we study the behavior of the Markov chain $X^{(N,r)}$ by exploring its connections to the deterministic process $y^{(r)}$.

We estimate the distance between $X^{(N,r)}$ and the set of maximal lotteries in several steps.
First, we choose $T_0$ large enough so that $y^{(r)}(\cdot,p_0)$ is close to $p^{(r)}$ for all but a small fraction of the time interval $[0,T_0]$ for all initial states $p_0$.
In \Cref{lem:uniformbounded}, we show that if $N$ is large enough, $X^{(N,r)}$ approximately solves (the integral equation equivalent to) the differential equation~\eqref{eq:diffeq} with high probability on the interval $[0,T_0]$ for any initial state.
From this we conclude in \Cref{lem:uniformboundedcontinuous} that for large enough $N$, $X^{(N,r)}$ is close to $y^{(r)}$ with high probability on any interval of length $T_0$, provided both processes start with the same state at the beginning of that interval.
Thus, $X^{(N,r)}$ is with high probability approximately equal to $p^{(r)}$ for all but a small fraction of rounds in any interval of length $T_0$.
Now we chop up the time line into successive intervals of length $T_0$.
In expectation, $X^{(N,r)}$ stays close to $y^{(r)}$ in a large fraction of these intervals. 
Using an adaption of the strong law of large numbers, we show in \Cref{lem:averagecongervence} that $X^{(N,r)}$ is \emph{almost surely} close to $p^{(r)}$ for all but a small fractions of rounds. 
Lastly, since by \Cref{lem:stationary}, $p^{(r)}$ is close to a maximal lottery if $r$ is small enough, \Cref{thm:main} follows.

The integral equation equivalent to \eqref{eq:diffeq} is
\begin{align}
	\begin{aligned}
		y(t) - y(0) &= \int_{0}^t f^{(r)}(y(s))ds\\
		y(0) &= p_0
	\end{aligned}
	\label{eq:inteq}
\end{align}
We show that $X^{(N,r)}$ approximately satisfies \eqref{eq:inteq} (with the integral replaced by a sum) for large $N$ on bounded time intervals.
\Cref{lem:uniformbounded} below states that for any time $T$ and any $\delta>0$, we can choose $N$ large enough so that with high probability, $X^{(N,r)}(n,p_0)$ does not violate \eqref{eq:inteq} by more than $\delta$ within the first $NT$ rounds independently of the initial state $p_0\in\Delta^{(N)}$. 
For the proof, we use the following proposition due to \citet[][Proposition 4.1]{Kurt70a}.
(The statement is adapted to our setting.)

\begin{proposition}[\citealp{Kurt70a}]\label{prop:kurtz}
	Let $(z^{(N)})_{N\in\mathbb N}$ be a sequence of discrete-time Markov chains with states spaces $A^{(N)}$ and probability transition matrices $Q^{(N)}$.
	Suppose there exist sequences of positive number $(\alpha_N)$ and $(\epsilon_N)$,
	$$
	\lim_{N\rightarrow\infty} \alpha_N = \infty \qquad\text{and}\qquad \lim_{N\rightarrow\infty}\epsilon_N = 0
	$$
	such that
	\begin{align}
		\sup_{N\in\mathbb N}\sup_{p\in A^{(N)}} \alpha_N\sum_{q\in A^{(N)}} \lvert p - q\rvert Q^{(N)}(p,q) <\infty
		\label{eq:kurtz1}
	\end{align}
	and
	\begin{align}
		\lim_{N\rightarrow\infty}\sup_{p\in A^{(N)}} \alpha_N\sum_{q\in A^{(N)}, \lvert p - q\rvert > \epsilon_N} \lvert p - q\rvert Q^{(N)}(p,q) = 0\text.
		\label{eq:kurtz2}
	\end{align}
	Let 
	\begin{align*}
		G^{(N)}(p) = \alpha_N\sum_{q\in A^{(N)}} (q - p)Q^{(N)}(p,q)\text.
	\end{align*}
	Then, for every $\delta > 0$ and $T > 0$,
	\begin{align*}
		\lim_{N\rightarrow\infty}\sup_{p\in A^{(N)}} \prob{\sup_{n\le \alpha_n T} \left\lvert z^{(N)}(n) - z^{(N)}(0) - \sum_{k=0}^{n-1} \frac1{\alpha_N} G^{(N)}(z^{(N)}(k))\right\rvert > \delta\mid z^{(N)}(0) = p} = 0\text.
	\end{align*}
\end{proposition}

The following lemma applies this result to $(X^{(N,r)})_{N\in\mathbb N}$ for a fixed $r$.

\begin{lemma}\label{lem:uniformbounded}
	For every $T > 0$ and $\delta > 0$,
	\begin{align}
		\lim_{N\rightarrow\infty} \sup_{p \in \Delta^{(N)}} \prob{\sup_{n \le NT} \left\lvert X^{(N,r)}(n,p) - X^{(N,r)}(0,p) - \sum_{k = 0}^{n-1} f^{(r)}\left(X^{(N,r)}(k,p)\right)\right\rvert \ge \delta} = 0
		\label{eq:uniformboundeddiscretetime}
	\end{align}
\end{lemma}

\begin{proof}
	Recall that $P^{(N,r)}$ is the transition probability matrix of $X^{(N,r)}$.

	We apply \Cref{prop:kurtz} with $z^{(N)} = X^{(N,r)}$, $A^{(N)} = \Delta^{(N)}$, $Q^{(N)} = P^{(N,r)}$, $\alpha_N = N$, and $\epsilon_N = \frac2N$ and
	check \eqref{eq:kurtz1} and \eqref{eq:kurtz2}:
	\begin{align*}
		&\sup_{N \in\mathbb N} \sup_{p \in\Delta^{(N)}} N \sum_{q \in\Delta^{(N)}} \lvert p - q\rvert P^{(N,r)}(Np,Nq)\\
		= &\sup_{N \in\mathbb N} \sup_{p \in\Delta^{(N)}} N \sum_{i,j = 1}^d \frac1N \lvert e_i - e_j\rvert P^{(N,r)}(Np,Np - e_i + e_j) \le 2
	\end{align*}
	and 
	\begin{align*}
		\lim_{N\rightarrow\infty} \sup_{p \in \Delta^{(N)}} N \sum_{q\in\Delta^{(N)}\colon \lvert p - q\rvert > \frac2N} \lvert p - q\rvert P^{(N,r)}(Np,Nq) = 0
	\end{align*}
	Recalling the definition of $f^{(r)}$ shows that $G^{(N)} = f^{(r)}$ for all $N$.
	Hence, \eqref{eq:uniformboundeddiscretetime} follows.
\end{proof}

	Note that \Cref{lem:uniformbounded} does not use the full strength of \Cref{prop:kurtz} since $G^{(N)} = f^{(r)}$ is independent of $N$. 
	Recall from \Cref{rem:proofadaption1} that for the urn process without replacement, $f^{(r)}$ does depend on $N$.
	Hence, the additional flexibility of \Cref{prop:kurtz} is needed in that case.

Since we want to compare the discrete-time process $X^{(N,r)}$ to the continuous-time process $y^{(r)}$ solving \eqref{eq:diffeq}, it is convenient to turn $X^{(N,r)}$ into a continuous-time process.
To this end, let $\bar X^{(N,r)}(t,p) = X^{(N,r)}(\lfloor Nt\rfloor,p)$ for all $t \ge 0$ and $p\in\Delta$.
(That is, time is scaled by $\frac1N$.)
$\bar X^{(N,r)}$ is a right-continuous step function, which takes steps of length $\frac1N \lvert e_i - e_j\rvert = \frac2N$ and is constant on time intervals $[\frac kN,\frac{k+1}N)$.
Thus, as $N$ grows, the steps become smaller and appear in shorter intervals.
\Cref{lem:uniformbounded} shows that on any bounded time interval, $\bar X^{(N,r)}$ satisfies~\eqref{eq:inteq} up to some arbitrary error with high probability when $N$ is large enough.
That is, for every $T > 0$ and $\delta > 0$,
\begin{align}
	\lim_{N\rightarrow\infty} \sup_{p\in \Delta^{(N)}} \prob{\sup_{t \le T} \left\lvert\bar X^{(N,r)}(t,p) - \bar X^{(N,r)}(0,p) - \int_0^t f^{(r)}\left(\bar X^{(N,r)}(s,p)ds\right)\right\rvert \ge \delta} = 0
	\label{eq:uniformbounded}
\end{align}

In \Cref{lem:uniformboundedcontinuous}, we show that this implies that the trajectories of $y^{(r)}(\cdot,p)$ and $\bar X^{(N,r)}(\cdot,p)$ stay close to each other with high probability on a given bounded time interval for any initial state $p$ for large $N$.
Importantly for later use, the bound on the probability is uniform in $p$.

\begin{lemma}\label{lem:uniformboundedcontinuous}
	For every $T > 0$ and $\delta > 0$,
	\begin{align}
		\lim_{N\rightarrow\infty} \sup_{p \in\Delta^{(N)}} \prob{\sup_{t \le T} \left\lvert y^{(r)}(t,p) - \bar X^{(N,r)}(t,p)\right\rvert \ge \delta} = 0\text.
		\label{eq:uniformboundedcontinuous}
	\end{align}
\end{lemma}

\begin{proof}
	First observe that since $f^{(r)}$ is continuously differentiable on the compact space $\Delta$, there is $C \in\mathbb R_{\ge 0}$ such that $f^{(r)}$ is Lipschitz-continuous with constant $C$.
	Let $T > 0$, $\delta > 0$, and $p\in\Delta$.
	If $\sup_{t \le T} \lvert\bar X^{(N,r)}(t,p) - \bar X^{(N,r)}(0,p) - \int_0^t f^{(r)}(\bar X^{(N,r)}(s,p))ds\rvert < \epsilon$,
	 
	then for all $t\in[0,T]$,
	\begin{align*}
		\left\lvert y^{(r)}(t,p) - \bar X^{(N,r)}(t,p)\right\rvert &= \left\lvert y^{(r)}(t,p) - y^{(r)}(0,p) - \bar X^{(N,r)}(t,p) + \bar X^{(N,r)}(0,p)\right\rvert\\
		&< \epsilon + \int_0^t \left\lvert f^{(r)}\left(y^{(r)}(s,p)\right) - f^{(r)}\left(\bar X^{(N,r)}(s,p)\right)\right\rvert ds\\ 
		&\le \epsilon + C \int_0^t \left\lvert y^{(r)}(s,p) - \bar X^{(N,r)}(s,p)\right\rvert ds
	\end{align*}
	The first inequality follows from the assumption about $\bar X^{(N,r)}$ and the fact that $y^{(r)}$ satisfies \eqref{eq:inteq}.
	The second inequality uses the Lipschitz-continuity of $f^{(r)}$.
	We apply Gr\"onwall's inequality to conclude that
	\begin{align*}
		\sup_{t\le T} \left\lvert y^{(r)}(t,p) - \bar X^{(N,r)}(t,p)\right\rvert < \epsilon e^{CT} < \delta
	\end{align*}	
	for $\epsilon > 0$ small enough.
	Note that the choice of $\epsilon$ does not depend on $p = \bar X^{(N,r)}(0,p)$.
	
	By \eqref{eq:uniformbounded}, for every $\rho > 0$, we can find $N_0\in\mathbb N$ such that for every $N \ge N_0$,
	\begin{align*}
		\sup_{p\in\Delta^{(N)}} \prob{ \sup_{t \le T} \left\lvert\bar X^{(N,r)}(t,p) - \bar X^{(N,r)}(0,p) - \int_0^t f^{(r)}\left(\bar X^{(N,r)}(s,p)\right)ds\right\rvert \ge \epsilon} < \rho
	\end{align*}
	Hence, for all $N \ge N_0$,
	\begin{align*}
	 	\sup_{p \in\Delta^{(N)}} \prob{\sup_{t \le T} \lvert y^{(r)}(t,p) - \bar X^{(N,r)}(t,p)\rvert \ge \delta} < \rho
	\end{align*}
	Since $\rho$ was arbitrary, \eqref{eq:uniformboundedcontinuous} follows.
\end{proof}

The last tool, \Cref{lem:slln}, is in essence a one-sided strong law of large numbers for indicator random variables.
Instead of the usual assumption of i.i.d. random variables, it only assumes that the probability of each variable being $1$ is conditionally upper bounded.
The elegant proof was pointed out to us by an anonymous referee.
\begin{lemma}\label{lem:slln}
	Let $\alpha\in[0,1]$.
	Let $\{Z_n\colon n\in\mathbb N_0\}$ be indicator random variables and, for $n \ge 1$, $S_n = \sum_{k = 1}^n Z_k$.
	If $\prob{Z_1 = 1} \le \alpha$ and for all $n \ge 2$, $\prob{Z_n =1 \mid S_{n-1}} \le \alpha$, then
	\begin{align*}
		\prob{\limsup_{n\rightarrow\infty} \frac{S_n}n > \alpha} = 0\text.
	\end{align*}
\end{lemma}
\begin{proof}
	The proof uses the moment generating functions of the random variables involved.
	Let $Z$ be an indicator random variable with $\prob{Z = 1} = \alpha$.
	Then, for all $n \ge 2$ and $t > 0$,
	\begin{align*}
		\ev{e^{tS_n}} = \ev{e^{tZ_n}e^{tS_{n-1}}} \le \ev{e^{tZ}} \ev{e^{tS_{n-1}}},
	\end{align*}
	where the inequality follows from the assumption that $\prob{Z_n =1 \mid S_{n-1}} \le \alpha$ and Fubini's theorem.
	Repeated application of this inequality and the assumption $\prob{Z_1 = 1} \le \alpha$ give that for all $n \ge 1$ and $t > 0$,
	\begin{align*}
		\ev{e^{tS_n}} \le \ev{e^{tZ}}^n.
	\end{align*}
	
	Now let $\beta > \alpha$, $n\ge 1$, and $t > 0$.
	Then, by Markov's inequality,
	\begin{align*}
		\prob{S_n \ge \beta n} = \prob{e^{tS_n} \ge e^{t\beta n}} \le e^{-t\beta n} \ev{e^{tS_n}}.
	\end{align*}
	Combining the two preceding inequalities gives
	\begin{align*}
		\prob{S_n \ge \beta n} \le e^{-t\beta n}\ev{e^{tZ}}^n = \ev{e^{t(Z - \beta)}}^n.
	\end{align*}
	Since $\ev{Z - \beta} < 0$, we have for sufficiently small $t > 0$ that $\ev{e^{t(Z - \beta)}} < 1$.
	Hence, $\prob{S_n \ge \beta n}$ decays exponentially in $n$, and so
	\begin{align*}
		\sum_{n \ge 1} \prob{S_n \ge \beta n} < \infty.
	\end{align*}
	The Borel-Cantelli lemma thus gives that $\lim\sup_n \frac{S_n}n \le \beta$ almost surely.
	Since $\beta > \alpha$ was arbitrary, we conclude that $\lim\sup_n \frac{S_n}n \le \alpha$ almost surely as claimed.
\end{proof}

Putting together \Cref{lem:continuoustimeconvergence}, \Cref{lem:uniformboundedcontinuous}, and \Cref{lem:slln}, 
we show that $\bar X^{(N,r)}$ is almost surely close to $p^{(r)}$ most of the time for large enough $N$.
More precisely, for $S \subset \Delta^{(N)}$, let
\begin{align*}
	\bar s_t^{(N,r)}(S) = \frac1t \int_0^t \chi_S(\bar X^{(N,r)}(s))ds
\end{align*}
be the fraction of time $\bar X^{(N,r)}$ spends in a state in $S$.
For $\delta > 0$, we consider the fraction of time spent in $B_\delta(p^{(r)})$.
If $N$ is large enough, the limit of $\bar s_t^{(N,r)}(B_\delta(p^{(r)}))$ for $t$ to infinity is almost surely close to 1.

\begin{lemma}\label{lem:averagecongervence}
	Let $\delta,\tau > 0$.
	Then, for every $r > 0$, there is $N_0$ such that for all $N\ge N_0$ and $p_0\in\Delta^{(N)}$,
	\begin{align}
				\prob{\lim_{t\rightarrow\infty} \bar s_t^{(N,r)}(B_\delta(p^{(r)})) \ge 1-\tau} = 1
		\label{eq:averageconvergence}
	\end{align}
\end{lemma}

\begin{proof}
	Let $r > 0$.
	By \Cref{lem:continuoustimestaysclose}, we can find $\eta > 0$ such that
	\begin{align*}
		\sup\left\{\left\lvert y^{(r)}(t,p) - p^{(r)}\right\rvert\colon t\ge 0,\, p\in B_\eta(p^{(r)})\right\} < \frac\delta2
	\end{align*}	
	By \Cref{lem:continuoustimeconvergence}, we can find $T_1 > 0$ such that for all $T \ge T_1$,
	\begin{align*}
		\sup\left\{\left\lvert y^{(r)}(T,p) - p^{(r)}\right\rvert\colon p\in\Delta\right\} < \eta
	\end{align*}
	Let $T_0 = \frac2\tau T_1$.
	Note that $y^{(r)}$ is time-invariant, that is, $y^{(r)}(t,p) = y^{(r)}(t - t_0, y^{(r)}(t_0,p))$ for all $t\ge t_0\ge 0$.
	Combining these facts, it follows that for every $p\in\Delta$, the measure of $t\in[t_0,t_0 + T_0]$ for which $y^{(r)}(t,p)$ is in an $\frac\delta2$-ball around $p^{(r)}$ is at least $(1-\frac\tau2)T_0$.
	We may assume that $T_0$ is integral.

	By \Cref{lem:uniformboundedcontinuous}, there is $N_0\in\mathbb N$ such that for all $N \ge N_0$,
	\begin{align*}
		\sup_{p \in\Delta^{(N)}} \prob{\sup_{0\le t \le T_0} \left\lvert y^{(r)}(t,p) - \bar X^{(N,r)}(t,p)\right\rvert \ge \frac\delta2} < \frac\tau2\text.
	\end{align*}
		
	Now fix $N \ge N_0$ and $p\in\Delta^{(N)}$.
	We upper bound the fraction of time $\bar X^{(N,r)}$ is further than $\delta$ away from $p^{(r)}$. 
	To simplify notation, let $t_k = kT_0$ and $\bar x_k = \bar X^{(N,r)}(t_k,p)$.
	For $n \ge 1$, we calculate the expected number of intervals $[t_{k-1},t_k]$, $1\le k\le n$ so that $\lvert\bar x(t,p) - p^{(r)}\rvert \ge \delta$ for some $t\in [t_{k-1},t_k]$.
	Let $Z^{(N,r)}_k$ be the indicator variable for the event that $\bar X^{(N,r)}(t,p)$ and $y^{(r)}(t-t_{k-1},\bar x_{k-1})$ differ by at least $\frac\delta2$ on the time interval $[t_{k-1},t_k]$ given that both start at the point $\bar x_{k-1}$ at time $t_{k-1}$.
	So $Z^{(N,r)}_k$ is 1 if $\sup_{t_{k-1} \le t \le t_k} \lvert\bar X^{(N,r)}(t,p) - y^{(r)}(t-t_{k-1},\bar x_{k-1})\rvert \ge \frac\delta2$ and 0 otherwise.
	Notice that $\{Z^{(N,r)}_k\colon k \in\mathbb N_0\}$ satisfies the hypothesis of \Cref{lem:slln} with $\alpha = \frac\tau2$.\footnote{While the probability that $Z^{(N,r)}_n$ equals 1 may depend on $Z^{(N,r)}_k$ for $k < n$, the bound of $\frac\tau2$ holds independently of the $Z^{(N,r)}_k$ since the bound obtained in \Cref{lem:uniformboundedcontinuous} is uniform in the initial state $p$.}

	If $Z^{(N,r)}_k = 0$, then
	\begin{align*}
		\int_{t_{k-1}}^{t_k}\bar\chi_{B_{\delta}(p^{(r)})}(\bar X^{(N,r)}(t,p))dt &\le\int_{t_{k-1}}^{t_k}\bar\chi_{B_{\frac\delta2}(y^{(r)}(t - t_{k-1},\bar x_{k-1}))}(\bar X^{(N,r)}(t,p))dt\\ 
		&+ \int_{t_{k-1}}^{t_k}\bar\chi_{B_{\frac\delta2}(p^{(r)})}y^{(r)}(t - t_{k-1},\bar x_{k-1}) dt\\
		&\le\frac\tau2T_0\text.
	\end{align*}
	It follows that
	\begin{align*}
		\frac1{nT_0} \sum_{k\in[n]} 		\int_{t_{k-1}}^{t_k}\bar\chi_{B_{\delta}(p^{(r)})}(\bar X^{(N,r)}(t,p))dt 
		&= \sum_{\substack{k\in[n]\\ Z^{(N,r)}_k = 0}} \frac1{nT_0} \int_{t_{k-1}}^{t_k}\bar\chi_{B_{\delta}(p^{(r)})}(\bar X^{(N,r)}(t,p))dt \\
		&+ \sum_{\substack{k\in[n]\\ Z^{(N,r)}_k = 1}} \frac1{nT_0}\int_{t_{k-1}}^{t_k}\bar\chi_{B_{\delta}(p^{(r)})}(\bar X^{(N,r)}(t,p))dt\\
		&\le \frac\tau2 + \frac1n\sum_{k = 1}^n Z^{(N,r)}_k\text.
	\end{align*}
	Applying \Cref{lem:slln} to $\{Z^{(N,r)}_k\colon k\in\mathbb N_0\}$ with $\alpha = \frac\tau2$ gives
	\begin{align*}
		\prob{\limsup_{n\rightarrow\infty} \frac1n\sum_{k=1}^n Z^{(N,r)}_k \ge \frac\tau2} = 0\text.
	\end{align*}
	Hence, with the preceding inequality we get
	\begin{align*}
		\prob{\lim_{t\rightarrow\infty}\bar s^{(N,r)}_t(B_\delta(p^{(r)}) \ge 1 - \tau} = \prob{\lim_{t\rightarrow\infty} \int_0^t \chi_{B_\delta(p^{(r)})}(\bar X^{(N,r)}) \ge 1 - \tau} = 1 
	\end{align*}
	which is \eqref{eq:averageconvergence}.
\end{proof}

\main*
\begin{proof}
	By \Cref{lem:stationary}, we can choose $r_0 > 0$ so that $p^{(r)} \in B_{\frac\delta2}(\ml(\p))$ for all $0 < r \le r_0$. 
	Let $0 < r \le r_0$ and $p^*\in\ml(\p)$ so that $\lvert p^{(r)} - p^*\rvert \le \frac\delta2$.
	Then, applying \Cref{lem:averagecongervence} to $\frac\delta2,\tau$, and $r$, we get $N_0\in\mathbb N$ such that \eqref{eq:averageconvergence} holds (with $\frac\delta2$ in place of $\delta$). 
	Hence, 
	\begin{align}
		\prob{\lim_{t\rightarrow\infty} \bar s_t^{(N,r)}(B_\delta(p^*)) \ge 1-\tau} = 1\text.
		\label{eq:almostsurelylikelyclose}
	\end{align}
	This is equivalent to the assertion in the first part of the theorem.
	
	The second statement follows by recalling the standard fact that the distribution of an irreducible and aperiodic Markov chain converges exponentially to its stationary distribution in the total variation norm \citep[][Theorem 4.9]{LPW09a}.
\end{proof}

\section{Approximate Axiomatics}\label{sec:approximateaxiomatics}

The temporal average of the urn process we describe approximates maximal lotteries. As mentioned in \Cref{sec:intro}, maximal lotteries are renowned for satisfying a large number of desirable properties and several results have shown that no other probabilistic social choice function (PSCF) satisfies these properties. 
For some axioms, it is obvious that every approximation of a PSCF that satisfies the axiom satisfies an approximate version of the axiom.
For example, every approximation of a Condorcet-consistent PSCF---one that returns the lottery with probability 1 on the Condorcet winner whenever one exists---is approximately Condorcet-consistent in the sense that it assigns probability close to 1 to Condorcet winners.
For other axioms, this is not true in general. 
In this section, we discuss population-consistency as an example and show that approximations of \emph{continuous} PSCFs that satisfy population-consistency satisfy an approximate consistency axiom.
Our notion of continuity requires that a preference change by a small fraction of the voters leads to at most a small change to the chosen lottery.
Since \ml is continuous and population-consistent, the urn procedure is approximately population-consistent. 
The same holds for several other axioms satisfied by \ml such as composition-consistency and agenda-consistency \citep[see][]{Bran13a}.
We do not discuss those here since the analysis is very similar.

The notation becomes more convenient by switching to \emph{fractional preference profiles}, which map every preference relation to the fraction of voters with these preferences. The set of all fractional preferences profiles is thus $\Delta(\mathcal R)$.
A PSCF is then a correspondence from $\Delta(\mathcal R)$ to $\Delta$.\footnote{This definition of a PSCF entails that it is anonymous (the identities of the voters are irrelevant) and homogeneous (replicating the entire electorate does not affect the outcome).}
Population-consistency is concerned with the consistency of PSCFs with respect to variable electorates: whenever a PSCF $\phi$ selects $p$ for two profiles on disjoint electorates, then it also selects $p$ for the union of both profiles. 
For fractional preference profiles, this is formalized as
\begin{align*}\textstyle
	\phi(\p') \cap \phi(\p'') \subset \phi(\frac12\,\p'+\frac12\,\p'') 
	\tag{population-consistency}
	\label{axiom:populationconsistency}
\end{align*}
for all $\p',\p''\in\Delta(\mathcal R)$.
An approximate version of population-consistency requires that $\phi(\p') \cap \phi(\p'')$ is contained in a small neighborhood of $\phi(\frac12\,\p'+\frac12\,\p'')$. 
For $\epsilon > 0$, $\phi$ satisfies $\epsilon$-approximate population-consistency at $\p\in\Delta(\mathcal R)$ if
\begin{align*}
	\phi(\p') \cap \phi(\p'') \subset B_\epsilon\left(\phi(\p)\right)
	\tag{$\epsilon$-approximate population-consistency}
	\label{axiom:approxpopulationconsistency}
\end{align*}
for all $\p',\p''\in\Delta(\mathcal R)$ with $\frac12\,\p' + \frac12\,\p'' = \p$.

There are population-consistent PSCFs and $\epsilon > 0$ such that even arbitrarily good approximations of the PSCF do not satisfy $\epsilon$-approximate population-consistency for some profile.
Call $\psi$ a $\delta$-approximation of $\phi$ for $\delta > 0$ if $\dist{\psi(\p),\phi(\p)} < \delta$ for all $\p$, where $\dist{U,V} = \inf_{p\in U,q\in V} \lvert p-q\rvert$ for $U,V\subset\Delta$.
If a PSCF is population-consistent, it may be the case that for some profile $\p$, for every $\delta > 0$, there is a $\delta$-approximation of the PSCF that violates $\epsilon$-approximate population-consistency for $\p$.
We show that this is impossible for (upper hemi-)continuous PSCFs at profiles for which a unique lottery is returned.\footnote{A PSCF $\phi$ is upper hemi-continuous if for all $\p\in\Delta(\mathcal R)$ and every sequence $(\p^k)_{k\in\mathbb N}\subset\Delta(\mathcal R)$ converging to $\p$, whenever a sequence $(p^k)_{k\in\mathbb N}\subset\Delta$ with $p^k\in\phi(\p^k)$ converges to $p\in\Delta$, then $p\in\phi(\p)$.}

\begin{restatable}{proposition}{approxpcons}\label{prop:approxpopulationconsistency}
	Let $\phi$ be a continuous PSCF that satisfies population-consistency. 
	Then, for every $\epsilon > 0$ and $\p\in\Delta(\mathcal R)$ with $\lvert \phi(\p)\rvert = 1$, there is $\delta > 0$ such that every $\delta$-approximation of $\phi$ satisfies $\epsilon$-approximate population-consistency at $\p$.
\end{restatable} 

\citet{Bran13a} have shown that $\ml$ is continuous and satisfies population-consistency.
Moreover, maximal lotteries are almost always unique in the sense that the set of all profiles with multiple maximal lotteries has measure zero and is nowhere dense in $\Delta(\mathcal R)$.\footnote{These facts follow from considering the dimension of the kernel of a sub-matrix of the skew-comparison matrix.}
\Cref{thm:main} implies that the number of balls in the urn and the mutation rate can be chosen so that the temporal average of the urn distribution is almost surely at most $\delta$ away from some maximal lottery. 
Hence, \Cref{prop:approxpopulationconsistency} implies that the urn process satisfies $\epsilon$-approximate population-consistency in a well-defined sense for arbitrary $\epsilon > 0$.

It is straightforward to check that \Cref{prop:approxpopulationconsistency} follows from the following lemma.

\begin{lemma}\label{lem:approxconvex}
	Let $\phi\colon U\rightrightarrows\mathbb R^d$ be an upper hemi-continuous correspondence from a compact subset $U\subset\mathbb R^k$ to $\mathbb R^d$
	Suppose that for all $u',u'' \in U$, $\phi(u') \cap \phi(u'') \subset \phi(\frac12\,u' + \frac12\,u'')$.
	Then, for all $u \in U$ and $\epsilon > 0$, there is $\delta > 0$ such that for all $u',u''\in U$ with $u = \frac12\,u' + \frac12\,u''$,
	\begin{align*}
		B_\delta(\phi(u')) \cap B_\delta(\phi(u'')) \subset B_\epsilon(\phi(u))\text.
	\end{align*}
\end{lemma}

\begin{proof}
	Let $u\in U$ and $\epsilon > 0$.
	Define
	\begin{align*}
		\delta' = \inf\left\{\dist{\phi(u') \setminus B_{\frac\epsilon2}(\phi(u)), \phi(u'') \setminus B_{\frac\epsilon2}(\phi(u)}\colon u',u'' \in U \text{ with } \frac12\, u' + \frac12\, u'' = u\right\}\text.
	\end{align*}
	(By convention, $\dist{\cdot,\cdot}$ takes the value $\infty$ if one of its arguments is the empty set.)
	Since $\phi$ is upper hemi-continuous, $U$ is compact, and $\phi(u') \cap \phi(u'') \subset \phi(u)$ for $u',u''$ as above, $\delta' > 0$. 
	
	Now let $\delta = \frac12\min\{\epsilon,\delta'\}$.
	Let $u',u''\in U$ with $\frac12\, u' + \frac12\, u'' = u$ and $p \in B_\delta(\phi(u')) \cap B_\delta(\phi(u''))$.
	Then there are $q'\in \phi(u')$ and $q'' \in \phi(u'')$ with $\lvert q' - p\rvert < \delta$ and $\lvert q'' - p\rvert < \delta$.
	Hence, $\lvert q' - q''\rvert < \delta'$.
	By definition of $\delta'$, either $q'\in B_{\frac\epsilon2}(\phi(u))$ or $q''\in B_{\frac\epsilon2}(\phi(u))$. 
	Either way, $p \in B_{\delta + \frac\epsilon2}(\phi(u)) \subset B_\epsilon(\phi(u))$, which is what had to be shown.
\end{proof}

\end{document}